%% file: 0-TC_2024.tex
\documentclass[10pt,journal]{IEEEtran}
\IEEEoverridecommandlockouts

\usepackage{multicol}
\usepackage{zlmtt}

\usepackage{pifont}
\usepackage{setspace}
\usepackage{color}
\usepackage{ulem}
\usepackage{microtype}
\usepackage{diagbox}
\usepackage{cite}
\usepackage{amsmath,amssymb,amsfonts}

\usepackage[linesnumbered,ruled,vlined]{algorithm2e}
\usepackage{listings}

\usepackage{enumitem}
\usepackage{subfigure}

\SetCommentSty{mycommfont}

\SetKwInput{KwInput}{Input}                
\SetKwInput{KwOutput}{Output}              
\SetKwInput{KwInitially}{Initially}

\usepackage{enumitem}
\usepackage{graphicx}
\usepackage{float}
\usepackage{textcomp}
\usepackage{xcolor}

\usepackage{multirow}
\usepackage{amsthm}
\usepackage{url}
\usepackage{threeparttable}
\usepackage{makecell}
\usepackage{subfigure}
\usepackage{amsmath}
\usepackage{mathrsfs}
\usepackage{bm}
\newcommand{\cmmnt}[1]{} 
\usepackage{lettrine}

\usepackage{booktabs} 
\usepackage[noend]{algpseudocode}

\usepackage{colortbl}
\usepackage{threeparttable}

\definecolor{colorwheel}{HTML}{006eb8}
\definecolor{brickred}{HTML}{7a2500}
\definecolor{mediumorchid}{HTML}{483a94}
\definecolor{applegreen}{HTML}{008000}

\usepackage{tikz}

\def\BibTeX{{\rm B\kern-.05em{\sc i\kern-.025em b}\kern-.08em
    T\kern-.1667em\lower.7ex\hbox{E}\kern-.125emX}}
\usepackage{algpseudocode}

\newtheorem{thm}{Theorem}

\usepackage{algpseudocode}
\usepackage{tikz,xcolor,hyperref}
\definecolor{lime}{HTML}{A6CE39}
\DeclareRobustCommand{\orcidicon}{
	\begin{tikzpicture}
	\draw[lime, fill=lime] (0,0)
	circle[radius=0.16]
	node[white]{{\fontfamily{qag}\selectfont \tiny \.{I}D}}; 
	\end{tikzpicture}	
	\hspace{-2mm}
}

\hypersetup{hidelinks} 

\begin{document}
\title{HiCoCS: High Concurrency Cross-Sharding on Permissioned Blockchains}

\author{Lingxiao~Yang,~\IEEEmembership{Graduate Student Member,~IEEE, }
	Xuewen~Dong,~\IEEEmembership{Member,~IEEE, }
	Zhiguo~Wan,~\IEEEmembership{Member,~IEEE, }
	Di~Lu,~\IEEEmembership{Member,~IEEE, }
	Yushu~Zhang,~\IEEEmembership{Senior Member,~IEEE, }
	and~Yulong~Shen,~\IEEEmembership{Member,~IEEE}
	\IEEEcompsocitemizethanks{
		\IEEEcompsocthanksitem Lingxiao Yang and Xuewen Dong (Corresponding author) are with the School of Computer Science and Technology, Xidian University, the Engineering Research Center of Blockchain Technology Application and Evaluation, Ministry of Education, and also with the Shaanxi Key Laboratory of Blockchain and Secure Computing, Xi’an 710071, China (e-mail: lxyang@stu.xidian.edu.cn; xwdong@xidian.edu.cn).
		\IEEEcompsocthanksitem Zhiguo Wan is with Zhejiang Lab, Hangzhou, Zhejiang 311121, China (e-mail: wanzhiguo@zhejianglab.com).
		\IEEEcompsocthanksitem Di Lu and Yulong Shen are with the School of Computer Science and Technology, Xidian University, and also with the Shaanxi Key Laboratory of Network and System Security, Xi’an 710071, China (e-mail: dlu@xidian.edu.cn; ylshen@mail.xidian.edu.cn).
		\IEEEcompsocthanksitem Yushu Zhang is with the College of Computer Science and Technology, Nanjing University of Aeronautics and Astronautics
		Nanjing, Jiangsu 210016, China (e-mail: yushu@nuaa.edu.cn).}
}

\maketitle
\begin{abstract}
As the foundation of the Web3 trust system, blockchain technology faces increasing demands for scalability. Sharding emerges as a promising solution, but it struggles to handle highly concurrent cross-shard transactions (\textsf{CSTx}s), primarily due to simultaneous ledger operations on the same account. Hyperledger Fabric, a permissioned blockchain, employs multi-version concurrency control for parallel processing. Existing solutions use channels and intermediaries to achieve cross-sharding in Hyperledger Fabric. However, the conflict problem caused by highly concurrent \textsf{CSTx}s has not been adequately resolved. To fill this gap, we propose HiCoCS, a high concurrency cross-shard scheme for permissioned blockchains. HiCoCS creates a unique virtual sub-broker for each \textsf{CSTx} by introducing a composite key structure, enabling conflict-free concurrent transaction processing while reducing resource overhead. The challenge lies in managing large numbers of composite keys and mitigating intermediary privacy risks. HiCoCS utilizes virtual sub-brokers to receive and process \textsf{CSTx}s concurrently while maintaining a transaction pool. Batch processing is employed to merge multiple \textsf{CSTx}s in the pool, improving efficiency. We explore composite key reuse to reduce the number of virtual sub-brokers and lower system overhead. Privacy preservation is enhanced using homomorphic encryption. Evaluations show that HiCoCS improves cross-shard transaction throughput by 3.5-20.2 times compared to the baselines.
\end{abstract}


\input{1-Introduction}

\input{1-1-RelatedWork}

\input{2-Preliminaries}

\input{3-SolutionOverview}

\input{4-Building}

\input{5-Security}
\input{6-PerformanceAnalysis}

\input{7-Conclusion}

\bibliographystyle{IEEEtran}
\normalem
\bibliography{cite}



\begin{IEEEbiography}[{\includegraphics[width=1in,height=1.25in,clip,keepaspectratio]{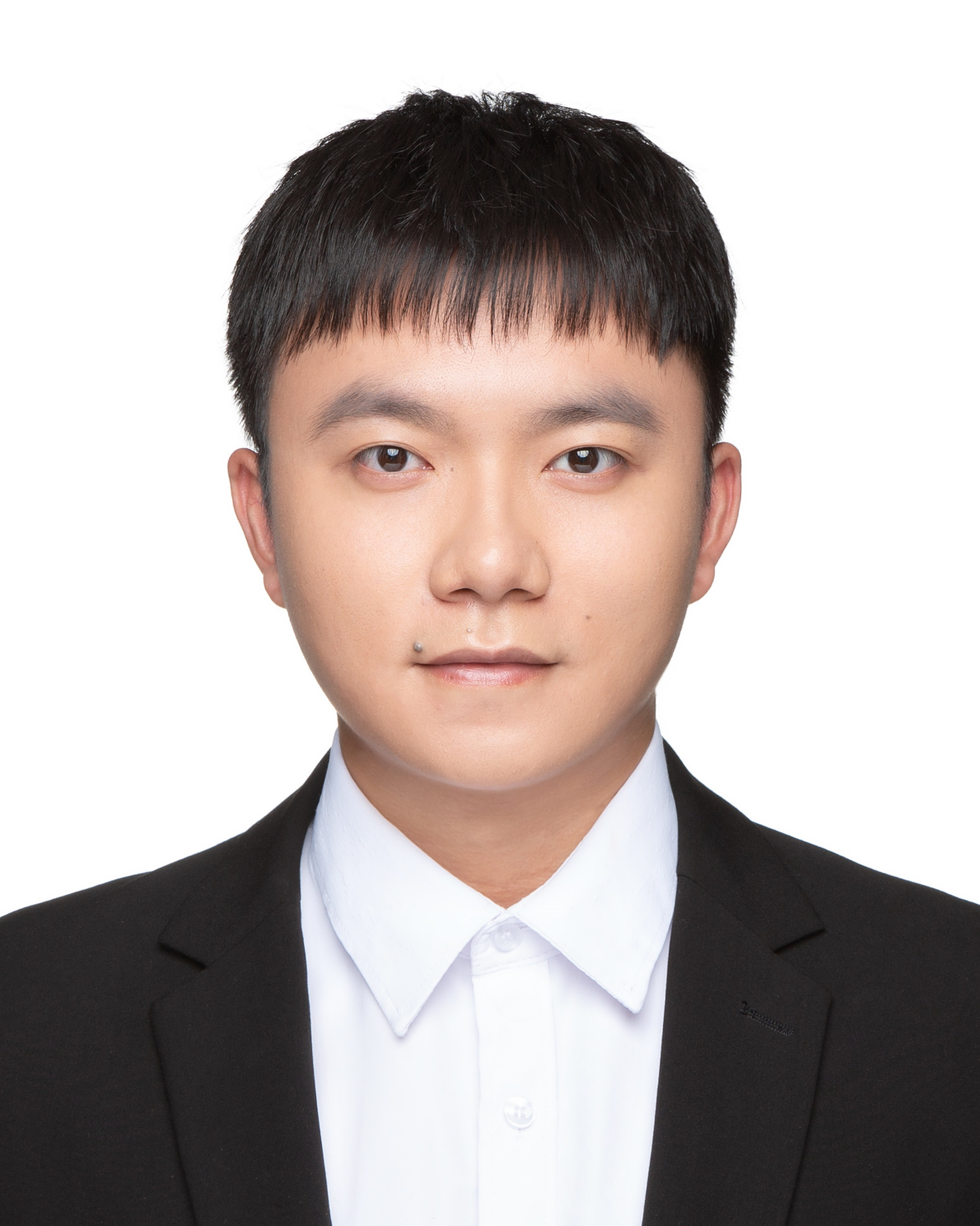}}]{Lingxiao Yang}
	(Graduate Student Member, IEEE) received B.E. degree in network engineering from Xidian University in 2018. He is a member of the Shaanxi Key Laboratory of Network and System Security. He is currently pursuing the Ph.D. degree with the School of Computer Science and Technology, Xidian University, Xi’an, China. His research interests include Web3 and blockchain applications.
\end{IEEEbiography}
\begin{IEEEbiography}[{\includegraphics[width=1in,height=1.25in,clip,keepaspectratio]{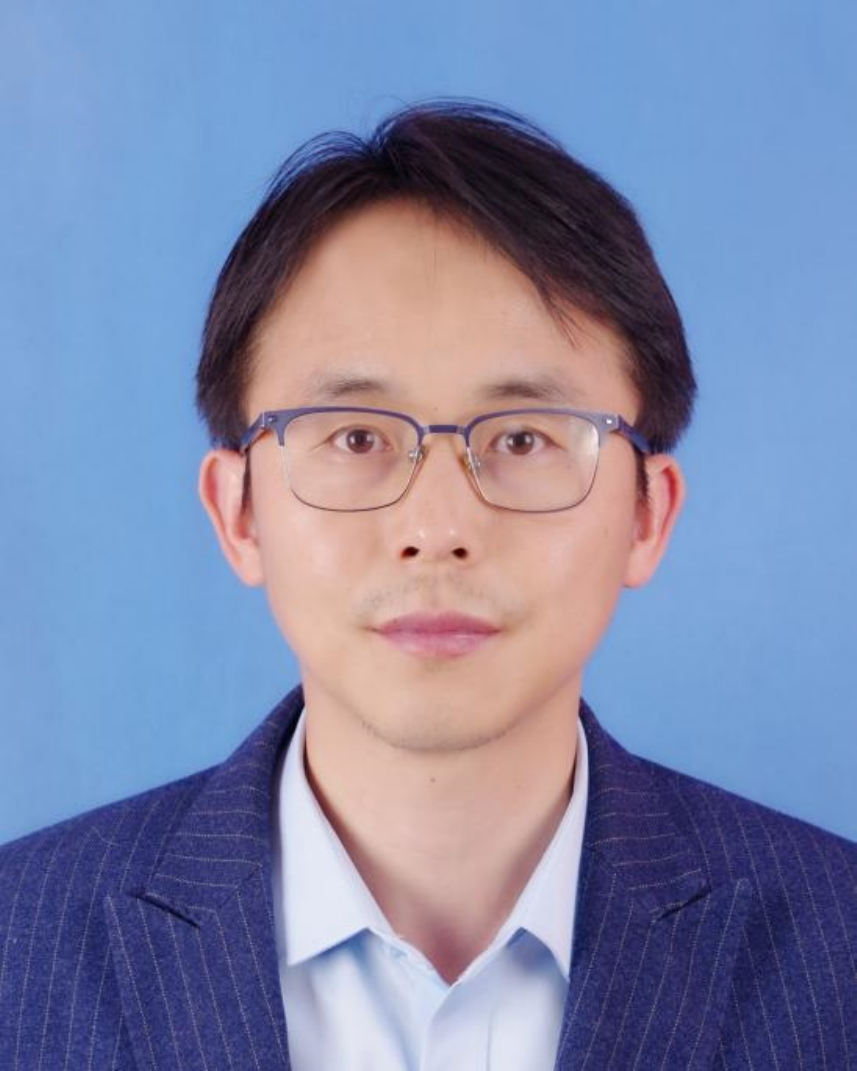}}]{Xuewen Dong}
	(Member, IEEE) received the B.E., M.S., and Ph.D. degrees in computer science and technology from Xidian University, Xi’an, China, in 2003, 2006, and 2011, respectively. From 2016 to 2017, he was with Oklahoma State University, Stillwater, OK, USA, as a Visiting Scholar. He is currently a Professor with the School of Computer Science and Technology, Xidian University. His research interests include cognitive radio network, wireless network security, and blockchain.
\end{IEEEbiography}

\begin{IEEEbiography}[{\includegraphics[width=1in,height=1.25in,clip,keepaspectratio]{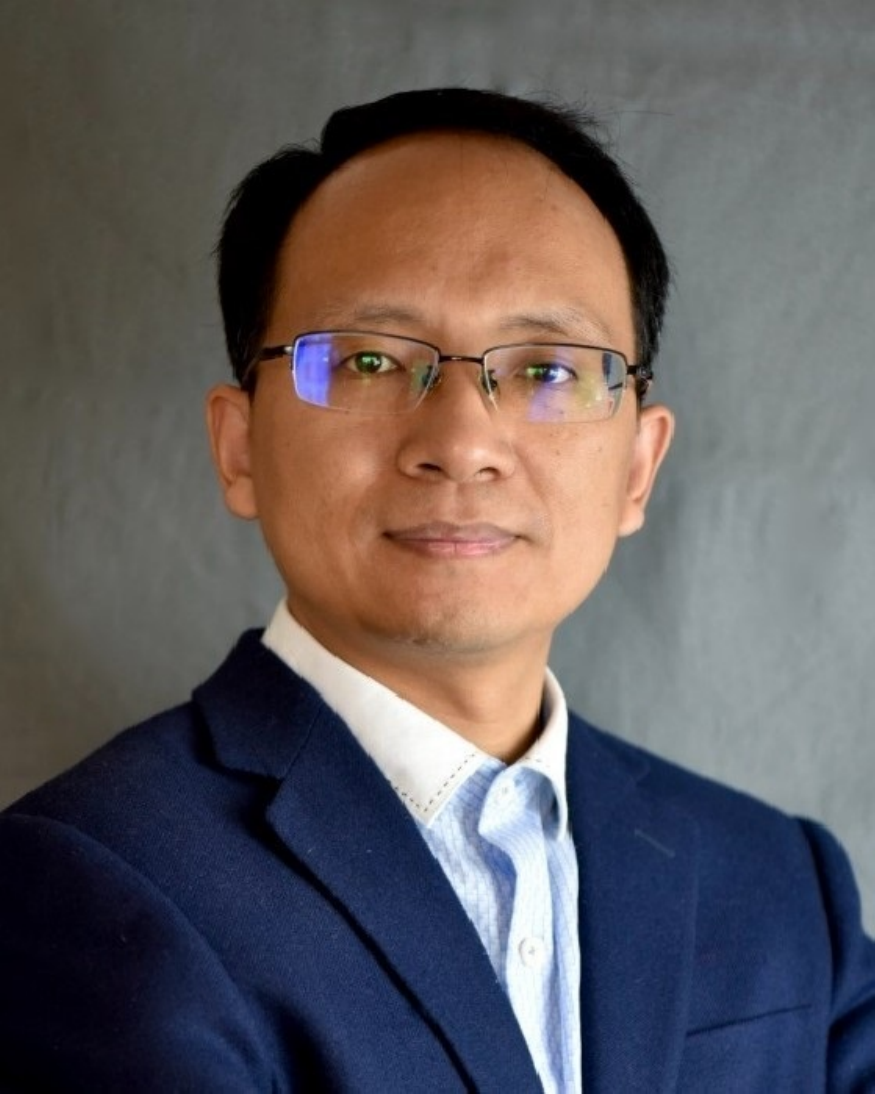}}]{Zhiguo Wan}
	(Member, IEEE) received the B.S. degree in computer science from Tsinghua University in 2002 and the Ph.D. degree in information security from the National University of Singapore in 2007. He is a Principal Investigator with the Zhejiang Laboratory, Hangzhou, Zhejiang, China. He was a Post-Doctoral Researcher with the Katholieke University of Leuven, Belgium; an Assistant Professor with Tsinghua University; and an Associate Professor with Shandong University, China. He has published more than 80 peer-reviewed conference and journal articles, including IEEE Transactions on Dependable and Secure Computing, IEEE Transactions on Information Forensics and Security, IEEE Transactions on Mobile Computing, and IEEE Transactions on Parallel and Distributed Systems. His main research interests include security and privacy for distributed systems, such as cloud computing, the Internet of Things, and blockchain. He has been a PC Member of dozens of international conferences, including IEEE INFOCOM.
\end{IEEEbiography}	

\begin{IEEEbiography}[{\includegraphics[width=1in,height=1.25in,clip,keepaspectratio]{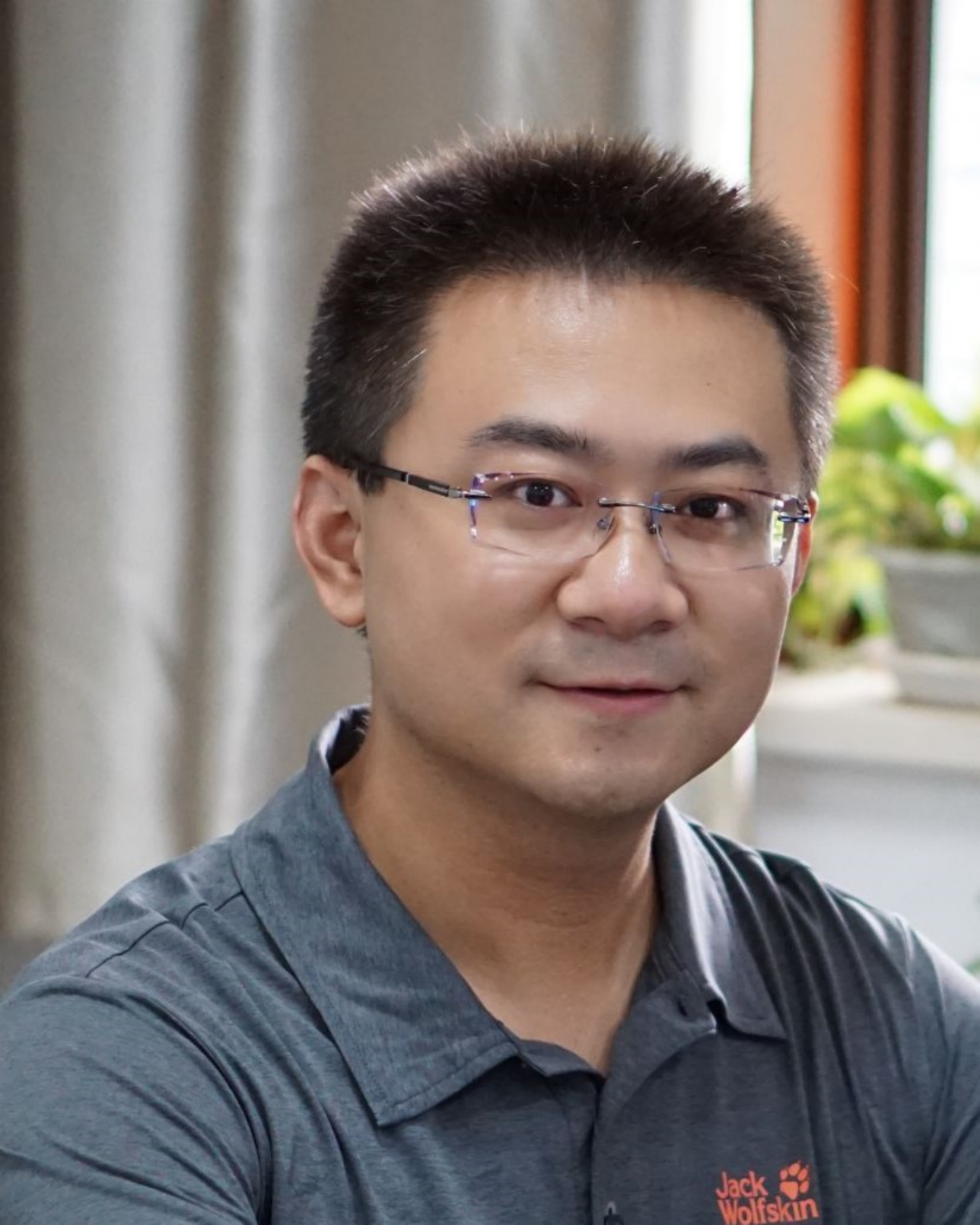}}]{Di Lu}
(Member, IEEE) received the BS, MS and PhD degrees in computer science and technology from Xidian University, China in 2006, 2009 and 2014. Currently, he is an associate professor with the school of Computer Science and Technology, Xidian University. His research interests include trusted computing, cloud computing, system and network security.
\end{IEEEbiography}

\begin{IEEEbiography}[{\includegraphics[width=1in,height=1.25in,clip,keepaspectratio]{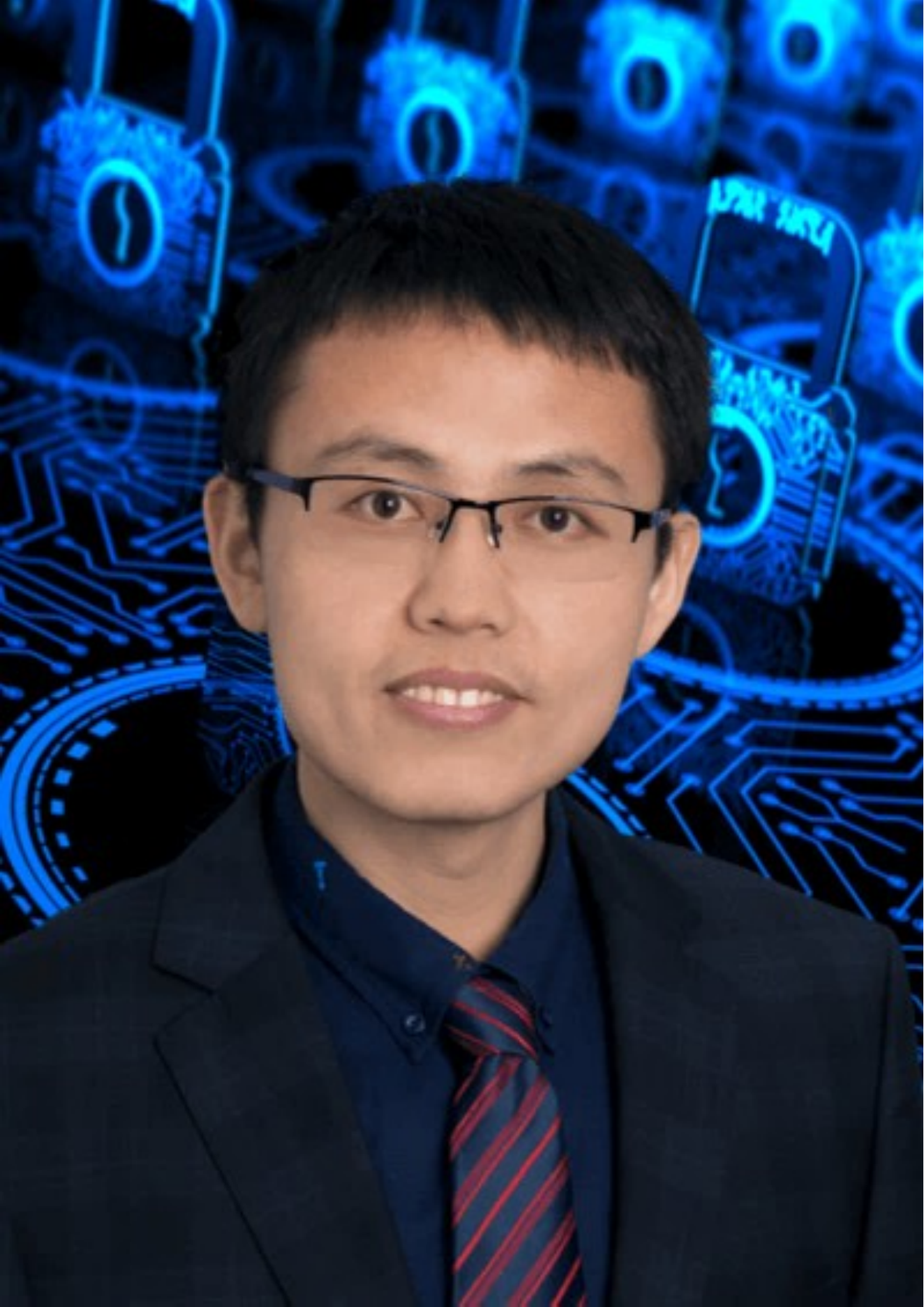}}]{Yushu Zhang}
	(Senior Member, IEEE) received the B.S. degree from the School of Science, North University of China, Taiyuan, China, in 2010, and the Ph.D. degree from the College of Computer Science, Chongqing University, Chongqing, China, 2014. He is currently a Professor with the College of Computer Science and Technology, Nanjing University of Aeronautics and Astronautics, China. He has produced more than 100 research peer-reviewed journal and conference papers. His research interests include multimedia security and blockchain. He is an Associate Editor of Information Sciences and Signal Processing.
\end{IEEEbiography}

\begin{IEEEbiography}[{\includegraphics[width=1in,height=1.25in,clip,keepaspectratio]{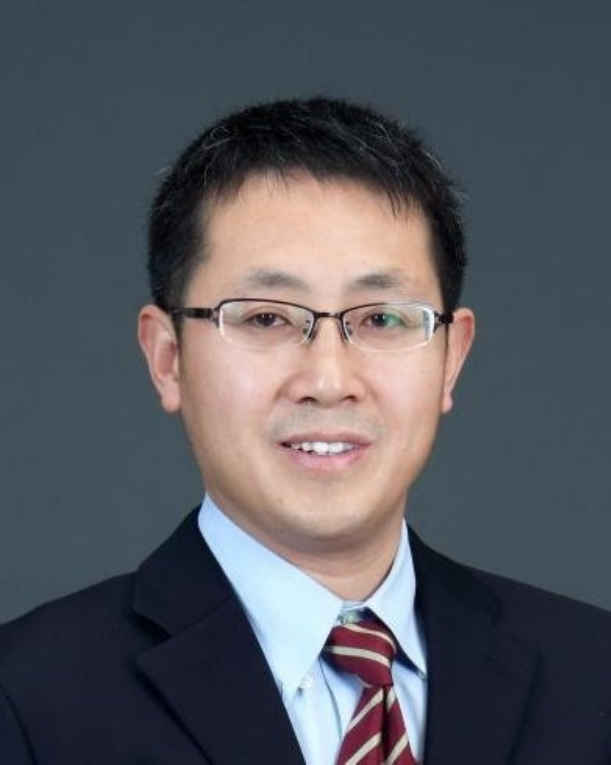}}]{Yulong Shen}
	(Member, IEEE) received the B.S. and M.S. degrees in computer science and the Ph.D. degree in cryptography from Xidian University, Xi’an, China, in 2002, 2005, and 2008, respectively. He is currently a Professor with the School of Computer Science and Technology, Xidian University, where he is also an Associate Director of the Shaanxi Key Laboratory of Network and System Security and a member of the State Key Laboratory of Integrated Services Networks. His research interests include wireless network security and cloud computing security. He has also served on the technical program committees of several international conferences, including ICEBE, INCoS, CIS, and SOWN.
\end{IEEEbiography}

\end{document}

%% file: 1-Introduction.tex
\section{Introduction} \label{intro}
\lettrine[lines=2]{T}{he} advent of Web3 marks a significant evolution in the Internet landscape, with blockchain technology playing a crucial role in establishing trust within this new paradigm \cite{zheng2017overview, zhu2021microthingschain, tong2023ti, yang2024asyncsc}. Despite its potential, blockchain technology faces significant scalability challenges, which limit its performance in large-scale applications \cite{yang2023optimal}. Sharding, a technique borrowed from distributed databases, has emerged as a promising solution to enhance blockchain scalability by partitioning the network into smaller, manageable segments that can process transactions independently, thereby increasing overall throughput \cite{huang2022brokerchain}.
	
Most existing blockchain sharding solutions\cite{luu2016secure, kokoris2018omniledger, zamani2018rapidchain, al2018chainspace, wang2019monoxide, Zicong21, huang2022brokerchain, licochain, androulaki2018hyperledger, tong2019hierarchical, hong2022scaling, cai2022benzene, peng2022neuchain, hong2023prophet, li2022achieving, liu2023flexible, hong2023gridb, qi2024lightcross, xu2024x, 10621394}, such as Elastico \cite{luu2016secure}, Omniledger \cite{kokoris2018omniledger}, Rapidchain \cite{zamani2018rapidchain}, and Pyramid \cite{Zicong21}, primarily focus on permissionless blockchains \cite{al2024sok}. However, permissioned blockchains like Hyperledger Fabric \cite{androulaki2018hyperledger}, which are widely used in enterprise settings, also require effective sharding solutions to meet the growing demand for scalable and efficient blockchain applications. MDIoTSP explores utilizing Fabric's channels as sub-shards to enhance scalability in multi-domain Internet of Things (IoT) systems \cite{tong2019hierarchical}. However, efficiently managing highly concurrent cross-shard transactions (\textsf{CSTx}s) remains a significant challenge. Existing solutions \cite{luu2016secure, kokoris2018omniledger, zamani2018rapidchain, al2018chainspace, wang2019monoxide, Zicong21, huang2022brokerchain, licochain, androulaki2018hyperledger, tong2019hierarchical, hong2022scaling, cai2022benzene, li2022achieving, liu2023flexible, hong2023gridb, qi2024lightcross, xu2024x, 10621394} frequently experience high transaction abort rates due to conflicts arising when multiple \textsf{CSTx}s attempt to access the same ledger state simultaneously. These conflicts are particularly problematic in permissioned blockchains, where maintaining high throughput and low latency is crucial for enterprise applications.

\begin{figure}[t]
	\vspace{-0.2cm}
	\subfigure[Illustration of the MVCC conflict.]{
		\begin{minipage}[t]{0.58\linewidth}
			\centering
			\includegraphics[width=2in]{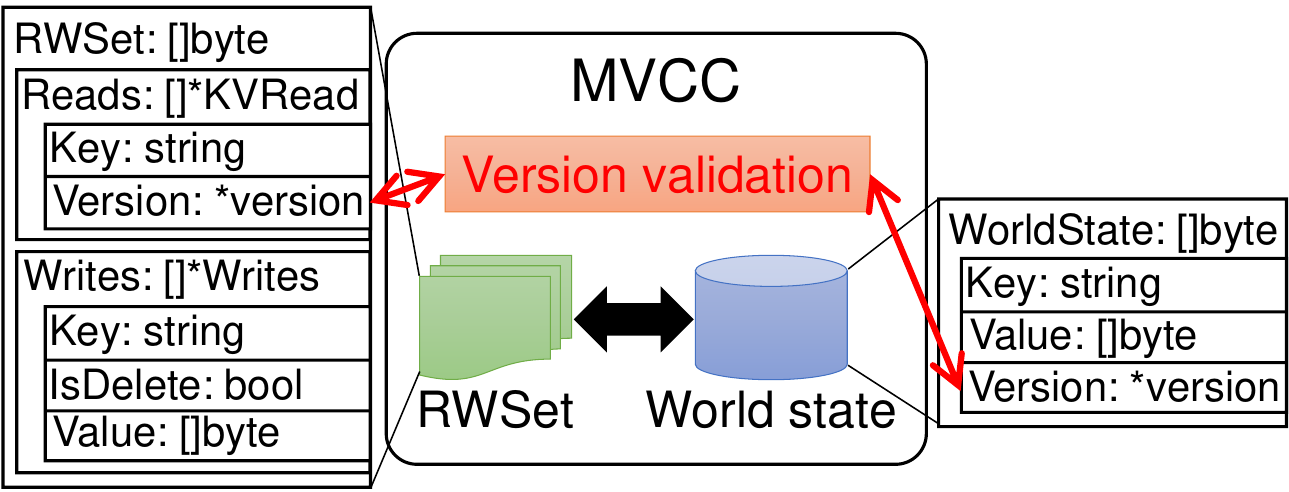}
			\label{1a}
			\vspace{-0.4cm}
		\end{minipage}
	}
	\subfigure[A vanilla version.]{
		\begin{minipage}[t]{0.35\linewidth}
			\centering
			\includegraphics[width=1.2in]{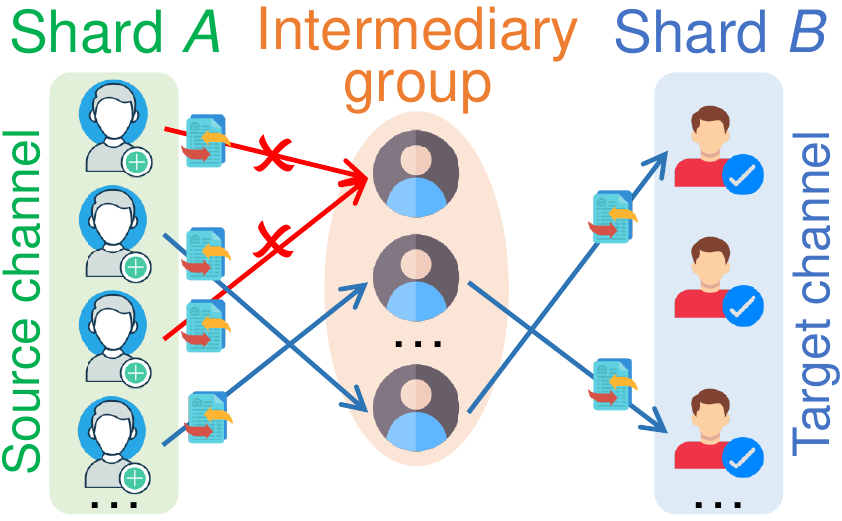}
			\label{1b}
			\vspace{-0.4cm}
		\end{minipage}
	}
	\centering
	\vspace{-0.4cm}
	\caption{The MVCC conflict and a cross-shard example in Hyperledger Fabric.}
	\vspace{-0.6cm}
\end{figure}

We further elucidate the existing problems. As shown in Fig. \ref{1a}, Fabric implements concurrent transaction data consistency based on multi-version concurrency control (MVCC) \cite{muro1984multi, wu2017empirical}. The concept is that each transaction's key-value pair data has a version number, and a new version is generated each time the data is modified. When a transaction reads or writes data, it checks whether the version number of the data is consistent with the expected value. If it is not, the transaction needs to be rolled back and re-executed. Thus, MVCC improves read concurrency but may increase transaction abort rates when handling write operations. Fig. \ref{1b} shows a \textit{vanilla version} of cross-shard implementation in Fabric using channels and intermediaries (or brokers \cite{huang2022brokerchain}). A source channel (i.e., shard \textit{A}) initiates \textsf{CSTx}s, and shard \textit{B} is a target channel that receives \textsf{CSTx}s. Each intermediary joins both shard \textit{A} and shard \textit{B}. Each \textsf{CSTx} needs to be transferred to an intermediary first, and then the intermediary transfers the assets to the target-shard user. As a result, when the concurrency of \textsf{CSTx}s is high, the ledger state of the intermediary account is frequently read and written, which causes MVCC conflicts due to data contention. The retries of aborted \textsf{CSTx}s increase overhead and degrade the quality of cross-shard service.

\textbf{Motivation.} As interoperability requirements expand, the volume of \textsf{CSTx}s increases rapidly, and mitigating the conflicts that often arise between them is key to improving the efficiency of \textsf{CSTx}s. With the adoption of sharding, inter-contract calls between contracts in different shards result in massive \textsf{CSTx}s. According to public data from Ethereum, more than 70\% of inter-contract calls occur more than twice \cite{Zicong21}. Large enterprises typically involve heavy partitioned transactions and data flows, which require \textsf{CSTx}s to manage data interactions between different partitions. Surprisingly, however, \textit{there is almost no research on high concurrency cross-sharding for Hyperledger Fabric.} In addition, traditional concurrency control mechanisms, such as two-phase locking (2PL) \cite{kokoris2018omniledger, al2018chainspace} and optimistic concurrency control (OCC) \cite{wang2019monoxide, Zicong21}, are ineffective at addressing these conflicts in high concurrency environments, leading to significant performance bottlenecks. Thus, to fill this gap, we conduct a systematic study to utilize Fabric's internal features to achieve high concurrency cross-sharding.

\textbf{Challenges.} We mainly face two challenges: (i) How to lightweight and effectively resolve the contradiction between the demand for highly concurrent \textsf{CSTx}s and MVCC conflicts among intermediaries? The naive approach is to directly add new intermediaries, but this increases resource consumption and system complexity. The added intermediaries still need to coordinate with other nodes when processing \textsf{CSTx}s, which may lead to new MVCC conflicts. (ii) How to balance concurrent transaction efficiency with privacy? The vanilla version overlooks the privacy-preserving issue of cross-sharding. In practice, intermediaries are semi-trusted. Each intermediary may honestly complete the \textsf{CSTx} processing, but it could be honest-but-curious, attempting to steal privacy (e.g., transaction amount) from data generated by the participants.

\textbf{Contributions.} To address the above challenges, we focus on mitigating the conflict problem of highly concurrent \textsf{CSTx}s and propose HiCoCS, a novel \textbf{hi}ghly \textbf{co}ncurrent \textbf{c}ross-\textbf{s}hard scheme. Different from prior work, HiCoCS introduces a unique virtual sub-broker mechanism using Fabric's composite key functionality. It only requires smart contract API calls, not the actual addition of new nodes, which greatly reduces resource consumption. By creating virtual sub-brokers on existing intermediary nodes, \textsf{CSTx}s can be refined to a smaller granularity, effectively reducing the frequency of MVCC conflicts. Batch management of virtual sub-brokers by intermediaries also improves transaction efficiency. On the other hand, HiCoCS combines homomorphic encryption to address the privacy preservation problem of transaction amounts calculated at the intermediaries. The use of composite keys efficiently identifies and manages \textsf{CSTx}s to avoid data conflicts, and homomorphic encryption allows operations on encrypted data without decryption, enabling the accumulation of concurrent \textsf{CSTx}s while maintaining data privacy and security. The combination of the two reduces conflicts and rollbacks in traditional concurrency control mechanisms. It also balances the security and performance of highly concurrent \textsf{CSTx}s. We summarize the contributions of HiCoCS as follows:

$\bullet$ \textbf{High concurrency cross-sharding.} To the best of our knowledge, HiCoCS is the first proposed solution to address the problem of highly concurrent \textsf{CSTx}s on Hyperledger Fabric. It utilizes composite keys to construct virtual sub-brokers for cross-shard intermediaries to mitigate concurrency conflicts, supporting the high concurrency of receiving and processing \textsf{CSTx}s. Based on maintaining a transaction pool, HiCoCS adopts the idea of batch processing to design an incremental accumulation module that summarizes the transactions in the pool, and the merged processing of multiple \textsf{CSTx}s reduces conflicts and improves efficiency.

$\bullet$ \textbf{Privacy-preserving cross-sharding.} For the privacy risk posed by semi-trusted intermediaries, we utilize Cheon-Kim-Kim-Song (CKKS) fully homomorphic encryption to achieve privacy preservation for \textsf{CSTx}s. It allows the incremental accumulation of \textsf{CSTx}s by intermediaries for virtual sub-brokers to be computed under ciphertext, balancing security and concurrency performance. We analyze and prove the security of the scheme.

$\bullet$ \textbf{Composite key reuse mechanism.} To further reduce the resource consumption caused by managing composite keys, we design a composite key reuse mechanism. It includes a Composite Key Proof of Equivalence (CKPoE) protocol that summarizes and regenerates composite keys to reduce the number of virtual sub-brokers and lower system overhead.

$\bullet$ \textbf{Prototype evaluation.} We implement a prototype of HiCoCS on Hyperledger Fabric and perform comprehensive comparisons with baselines. Evaluation results show that HiCoCS outperforms the state-of-the-art schemes in terms of transaction success rate (improved by 2.2 to 8.1 times), throughput (improved by 3.5 to 20.2 times), latency (reduced by 43.9\% to 62.0\%), and CPU \& memory utilization rates.

%% file: 1-1-RelatedWork.tex
\section{Related Work} \label{RW}
In this section, we categorize the literature related to \textsf{CSTx} conflict resolution based on its implementation and compare the corresponding features.

Existing work \cite{luu2016secure, kokoris2018omniledger, zamani2018rapidchain, al2018chainspace, wang2019monoxide, Zicong21, huang2022brokerchain, hong2023prophet, licochain, androulaki2018hyperledger, tong2019hierarchical, hong2022scaling, cai2022benzene, peng2022neuchain, li2022achieving, liu2023flexible, hong2023gridb, qi2024lightcross, xu2024x, 10621394} mainly employs two concurrency control mechanisms to mitigate \textsf{CSTx} conflicts: (i) \textit{Two-phase locking (2PL).} It is a pessimistic concurrency control approach that assumes transaction conflicts occur frequently, thus ensuring isolation by adding locks during transaction execution. (ii) \textit{Optimistic concurrency control (OCC).} It assumes that conflicts are relatively infrequent, does not add locks during execution, and verifies whether a conflict arises when the transaction is committed. Based on this, we illustrate representative sharded permissionless and permissioned blockchains, providing a feature comparison in Table \ref{features}.

\textit{2PL schemes:} Omniledger \cite{kokoris2018omniledger} proposes a Byzantine shard atomic commit protocol, Atomix, to ensure the consistency of transactions. Rapidchain \cite{zamani2018rapidchain} parallelizes data and computation through full node sharding. Chainspace \cite{al2018chainspace} proposes object-oriented smart contract sharding, which assigns different contracts and transactions to different shards for processing. Tong \textit{et al}. \cite{tong2019hierarchical} implemented a sharding system for Hyperledger Fabric using channels. AHL+ \cite{dang2019towards} relies on trusted hardware to enhance the Hyperledger Fabric to handle cross-shard distributed transactions. Aeolus \cite{zheng2022aeolus} proposes distributed state update sharding to maintain consistency across different clusters. Set \textit{et al}. \cite{set2022service} proposed a service-aware dynamic sharding that utilizes a reference committee to act as a coordinator for concurrency control.

\textit{OCC schemes:} Monoxide \cite{wang2019monoxide} achieves fast processing and eventual consistency of \textsf{CSTx}s through asynchronous cross-zone validation. Pyramid \cite{Zicong21} employs a hierarchical sharding architecture and a recursive consensus protocol to ensure efficient communication and global consistency across different layers and shards. CoChain \cite{licochain} designed a cross-shard Consensus on Consensus mechanism to securely configure small shards and enhance concurrency. Androulaki \textit{et al}. \cite{androulaki2018channels} achieved sharding on Hyperledger Fabric using multiple channels and utilized Merkle Tree to process \textsf{CSTx}s. Meepo \cite{zheng2022meepo} introduces cross-epoch and cross-call protocols for ordered cross-shard communication.

\textbf{Comparison.} Existing approaches have their own merits. However, none of these efforts properly address the conflict problem in highly concurrent \textsf{CSTx} scenarios. According to the evaluation of \textsc{Prophet} \cite{hong2023prophet}, more than half of the \textsf{CSTx}s are aborted or rolled back due to race conditions (i.e., frequent read/write to the same ledger state). Thus, 2PL and OCC schemes usually exhibit high transaction abort rates (i.e., aborted \textsf{CSTx}s may require multiple retries to succeed), sacrificing efficiency for the serializability of \textsf{CSTx}s. Our proposed HiCoCS uses composite keys to create virtual sub-brokers for conflicted intermediary accounts. It uses a message-passing approach to ensure the consistency of \textsf{CSTx}s, utilizing storage space in exchange for high concurrency \textsf{CSTx} efficiency. In addition, most existing work does not consider \textsf{CSTx} privacy preservation. HiCoCS combines homomorphic encryption to balance performance and security. Meanwhile, HiCoCS supports batch processing of \textsf{CSTx}s by utilizing composite keys for accumulated summarization of transaction ciphertexts. Although HiCoCS is a Hyperledger Fabric-oriented permissioned blockchain sharding scheme, it can also be extended to permissionless blockchains with simple modifications (see Section \ref{App-A3}).

\begin{table}[t]
	\centering
	\caption{Feature Comparison with Existing Sharding Solutions.}
	\vspace{-0.2cm}
	\resizebox{1\columnwidth}{!}{
		\begin{threeparttable}
			\begin{tabular}{lcccccc}
				\toprule
				\diagbox{Schemes}{Features$^*$} & \makecell{CCM} & \makecell{BT} & \makecell{Fabric- \\ oriented} & \makecell{HiCo \\ \textsf{CSTx}s} & \makecell{PP \\ \textsf{CSTx}s} & \makecell{BP \\ \textsf{CSTx}s} \\
				\midrule
				Ref. \cite{kokoris2018omniledger, zamani2018rapidchain} & 2PL & Pl & - & \ding{55}  & \ding{55} & \ding{55} \\
				Chainspace \cite{al2018chainspace} & 2PL & Pl & - & \ding{55}  & \ding{51} & \ding{55} \\
				Ref. \cite{tong2019hierarchical, dang2019towards, set2022service} & 2PL & Pd & \ding{51} & \ding{55}  & \ding{55} & \ding{55} \\
				Aeolus-Geth \cite{zheng2022aeolus} & 2PL & Pd & - & \ding{55}  & \ding{55} & \ding{55} \\
				Ref. \cite{wang2019monoxide, Zicong21, licochain} & OCC & Pl & - & \ding{55}  & \ding{55} & \ding{55} \\
				Ref. \cite{androulaki2018channels} & OCC & Pd & \ding{51} & \ding{55}  & \ding{55} & \ding{55} \\
				Meepo-OE \cite{zheng2022meepo} & OCC & Pd & - & \ding{55}  & \ding{55} & \ding{55} \\
				HiCoCS (This work) & CK & Pd/Pl & \ding{51} & \ding{51}  & \ding{51} & \ding{51} \\
				\bottomrule
			\end{tabular}
			\begin{tablenotes}
				\item[*] Abbreviation explanation: Concurrency Control Mechanism (CCM); Blockchain Type (BT); Highly Concurrent (HiCo); Privacy-Preserving (PP); Batch Processing (BP); Permissionless (Pl); Permissioned (Pd); Composite Key (CK). \ding{51}: Satisfied; \ding{55}: Unsatisfied; -: Inapplicable.
			\end{tablenotes}
		\end{threeparttable}
	}
	\label{features}
	\vspace{-0.5cm}
\end{table}

%% file: 2-Preliminaries.tex
\section{Background and Preliminaries} \label{pre}
\subsection{EOV Transaction Processing Architecture} \label{App-A1}
Unlike the order-execute-validate (OEV) transaction processing flow adopted by permissionless blockchains (e.g., Bitcoin and Ethereum), permissioned blockchains (e.g., Hyperledger Fabric) use an execute-order-validate (EOV) architecture. Fig. \ref{2} shows the transaction processing in Fabric: \ding{202} Clients initiate transaction proposals. \ding{203} Peer nodes simulate the execution of transactions (endorsements). \ding{204} Clients compare the endorsement results. \ding{205} The endorsement results are sent to the ordering service as transaction messages. \ding{206} The ordering nodes order and pack the transactions into blocks. \ding{207} The new blocks are broadcasted and validated in the peer cluster. \ding{208} The peer nodes append the block to its channel's ledger.

Transactions can be executed concurrently on multiple nodes during the simulation stage, while the correct order of transaction submission is determined in the ordering and validation stages. This approach performs well with few read/write conflicts and low transaction submission latency.

\begin{figure}[t]
	\centering
	\includegraphics[width=3in]{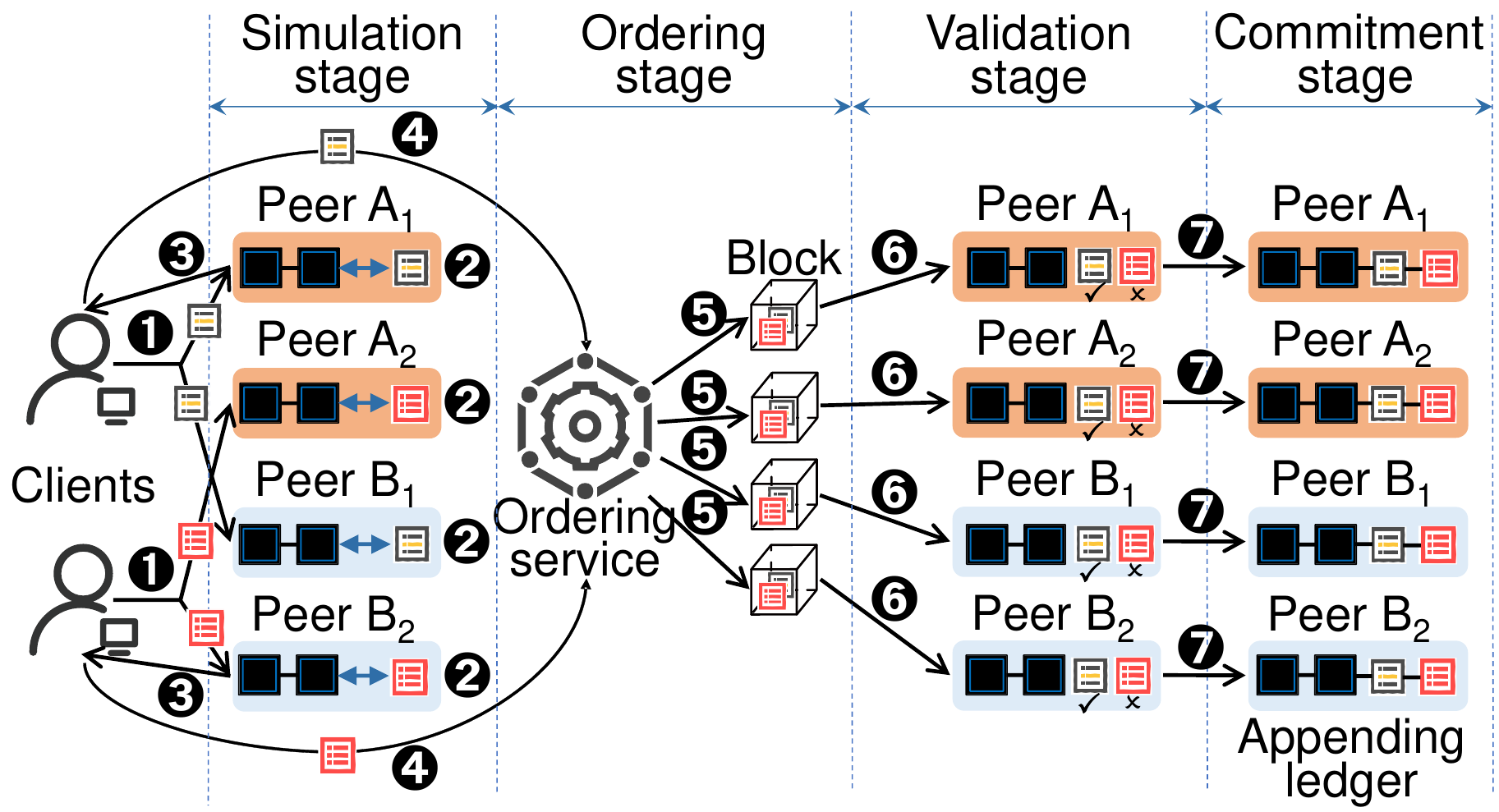}
	\vspace{-0.3cm}
	\caption{Hyperledger Fabric's transaction processing flow.}
	\label{2}
	\vspace{-0.5cm}
\end{figure}

\subsection{Multi-Version Concurrency Control (MVCC)}
In Hyperledger Fabric, MVCC defines the version number as the height of a transaction and stores it in the world state along with the key-value pair. The version of a key is contained only in the read set, which is used to check the validity of a transaction. The write set is used to update the value corresponding to a key, which is incremented after each successful update. This ensures transaction isolation; read and write operations do not interfere with each other, and conflict detection occurs only at the commitment stage. If the key version in the read set matches that in the world state, the transaction is valid; otherwise, it is considered an MVCC conflict.

Specific conflict cases include: (i) \textit{Write-write conflict.} If two concurrent transactions attempt to modify the same data item, only the first committed transaction succeeds, and the others abort due to the conflict. Because the data operated on by each transaction is a ``snapshot", concurrent transactions cannot see the changes made by other transactions and may make decisions based on outdated data. (ii) \textit{Repeatable read problem.} If a transaction reads a data item and then another transaction modifies it, the data has changed when the original transaction reads it again, which may lead to an abort because it needs to make a decision based on the latest data. (iii) \textit{Validation phase abort.} In MVCC-enabled systems (e.g., Oracle database, Ethereum blockchain, and other OCC policy-enabled systems), conflict detection is performed when the transaction commits, and the transaction is aborted if a conflict is found.

\subsection{Composite Key in Hyperledger Fabric}
The composite key in Hyperledger Fabric is a data structure used in state databases (e.g., CouchDB or LevelDB). It consists of multiple attributes/fields combined into a unique key. It supports efficient retrieval based on different combinations of attributes, enhancing the flexibility and complexity of business logic \cite{Sismanis2006}. The main steps of composite key usage in Fabric include: (i) \textit{Creation.} Call the \verb|CreateCompositeKey()| function of the chaincode API\footnote{https://hyperledger.github.io/fabric-chaincode-node/release-2.4/api/fabric-shim.ChaincodeStub.html} to concatenate the prefix and attribute array to generate a composite key. (ii) \textit{Store.} Call the \verb|PutState()| function to store the composite key and the corresponding value in the database. (iii) \textit{Query.} Call the \verb|GetStateByPartialCompositeKey()| function to query the data in the database based on a given partial composite key.

Composite keys are mainly used in Fabric for chaincode querying and logical indexing of state data, without directly affecting the physical storage structure of the state database. This provides an opportunity to construct composite keys for intermediaries to process transactions as virtual sub-brokers and share the pressure of \textsf{CSTx} conflicts when realizing Fabric's high concurrency cross-sharding.

\subsection{Homomorphic Encryption Techniques}
Homomorphic encryption is a public key encryption algorithm designed to perform data operations without leaking private data. Based on their arithmetic capabilities, homomorphic encryption schemes are classified into three categories \cite{acar2018survey}: (i) Partially homomorphic encryption (PHE) \cite{yaji2018privacy}; (ii) Somewhat homomorphic encryption (SHE) \cite{shoukry2016privacy}; (iii) Fully homomorphic encryption (FHE) \cite{lee2022privacy}. PHE supports only one of additive or multiplicative homomorphic operations, and the number of operations is unlimited. SHE supports both additive and multiplicative homomorphisms with a limited number of operations. FHE supports not only additive and multiplicative homomorphisms but also an unlimited number of operations.

In the vanilla version of Fabric's cross-sharding, a semi-trusted intermediary can batch process non-integer cross-shard assets, thus requiring privacy computations that support addition and multiplication. This paper uses an approximate FHE algorithm, i.e., the Cheon-Kim-Kim-Song (CKKS) \cite{cheon2017homomorphic}, which supports high-performance fully homomorphic encryption operations on floating-point numbers.

%% file: 3-SolutionOverview.tex
\begin{figure}[t]
	\centering
	\includegraphics[width=3in]{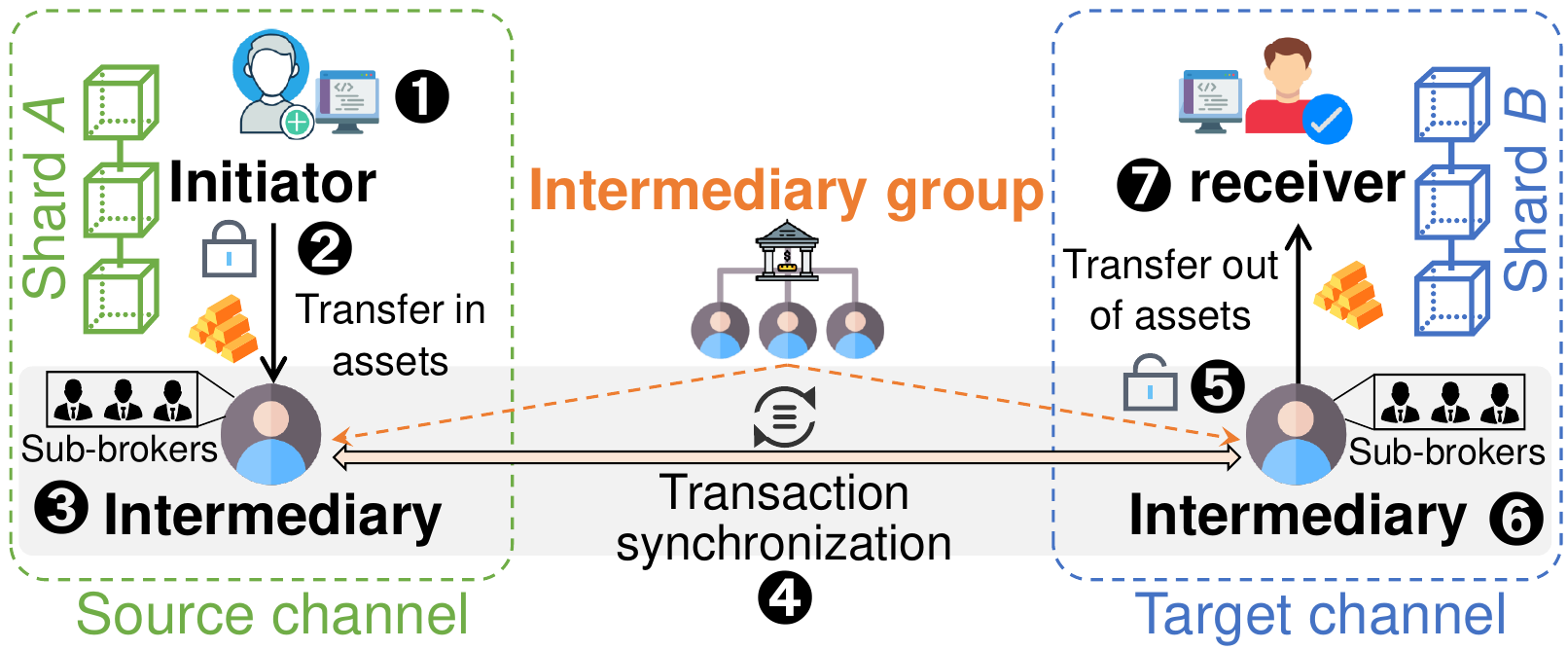}
	\vspace{-0.3cm}
	\caption{System overview of HiCoCS.}
	\label{3}
	\vspace{-0.3cm}
\end{figure}

\section{System Overview of HiCoCS} \label{solu}
\subsection{System Model} \label{System}
Fig. \ref{3} shows the overall interaction model of HiCoCS. The system roles include the transaction \textit{initiator} in the source channel, the \textit{intermediary group}, and the transaction \textit{receiver} in the target channel. Their capabilities are as follows.

\begin{itemize}[leftmargin=*]
	
	\item[$\bullet$] \textbf{Initiator.} As an initiator user of \textsf{CSTx}s, it first transfers assets in the source channel to an intermediary account.
	
	\item[$\bullet$] \textbf{Intermediary group.} It consists of multiple intermediary nodes, each of which should be satisfied with owning assets on both the source and target channels. An intermediary summarizes and processes multiple \textsf{CSTx}s in a unified manner to realize cross-shard asset transfers between the source and target channels, i.e., it receives the source channel's transactions and distributes them to the target channel. In the system's building block perspective, each intermediary performs  \textsf{CSTx} pre-processing (including the construction of \textit{composite keys/virtual sub-brokers}), incremental accumulation (calculated under ciphertext), transaction synchronization, and composite key reuse.
	
	\item[$\bullet$] \textbf{Receiver.} As a user on the target channel, it receives the assets from the intermediary for each \textsf{CSTx}.
\end{itemize}

The cross-shard trading process of HiCoCS is mainly divided into the following steps.

\begin{enumerate}[leftmargin=*]
	\item \textbf{\textsf{CSTx} initiation.} \ding{202} First, the initiator initiates an asset transfer request through its client. The transaction information contains the initiator's address (key), the receiver's address, and the transaction amount. \ding{203} Then, the initiator's \textsf{CSTx} message is encrypted and sent to an intermediary.
	
	\item \textbf{Intermediary processing.} \ding{204} The system transfers the initiator's transfer assets to an intermediary account (i.e., its virtual sub-broker accounts) in the source channel through the smart contract. \ding{205} The intermediary of the source channel collects the transaction ciphertexts, converts them into CKKS ciphertext vectors, and accumulates and summarizes the transaction ciphertexts. Finally, the target channel intermediary synchronizes the summarized results.
	
	\item \textbf{\textsf{CSTx} completion.} \ding{206} The source channel calls a decryption function of smart contract to restore the intermediary's cipher state summarized results to plaintext and complete the final accumulated amount to be written to the target channel's intermediary ledger. \ding{207} Then, the intermediary in the target channel first cross-channel queries the set of successfully received transactions on the source channel as its set of transactions to be transferred. \ding{208} Finally, it traverses the pending set and transfers the assets to the final target account in proportion to an exchange rate.
\end{enumerate}

\begin{figure}[t]
	\centering
	\includegraphics[width=3in]{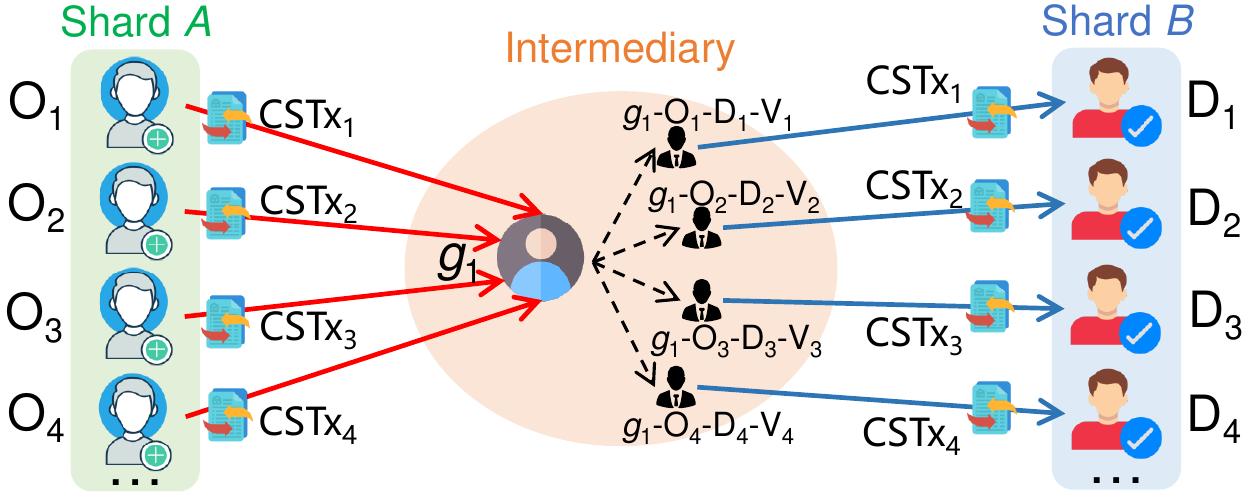}
	\vspace{-0.3cm}
	\caption{An example illustrating the HiCoCS.}
	\label{F12}
	\vspace{-0.3cm}
\end{figure}

We give a concise example to illustrate the core idea of HiCoCS in handling high concurrency \textsf{CSTx}s, as shown in Fig. \ref{F12}. Suppose there are multiple initiators $\{\textsf{O}_1, \textsf{O}_2, \textsf{O}_3, \textsf{O}_4\}$ initiating $\{\textsf{CSTx}_1, \textsf{CSTx}_2, \textsf{CSTx}_3, \textsf{CSTx}_4\}$ to the receivers $\{\textsf{D}_1, \textsf{D}_2, \textsf{D}_3, \textsf{D}_4\}$ respectively in a short period of time. These \textsf{CSTx}s are all relayed and handled by an intermediary $g_1$. In the vanilla version, most of the \textsf{CSTx}s will be aborted due to MVCC conflicts, as the ledger state of $g_1$ is read and written multiple times. However, in HiCoCS, the intermediary generates a virtual sub-broker for each \textsf{CSTx}, i.e., invokes the smart contract to quickly generate the composite keys $\{g_1-\textsf{O}_i-\textsf{D}_i-\textsf{V}_i\}|_{i=1}^{4}$. The $\textsf{V}_i $ is the amount's ciphertext for each \textsf{CSTx}. Since multiple virtual sub-brokers are writing to the ledger, the MVCC conflicts of a single intermediary are effectively mitigated. The overhead of creating a composite key is much smaller than maintaining a new intermediary, thus significantly reducing the transaction cost for the users.

\textbf{Insight.} The core principle of HiCoCS support for high concurrency \textsf{CSTx} processing is to \textit{build multiple virtual sub-brokers through the composite keys to avoid concurrent conflicts and realize batch processing of transactions. Each intermediary then counts the final amount by querying the pool of composite key transactions, and multiple \textsf{CSTx}s are processed at once} (see Section \ref{Building}). Actually, the inspiration originally came from blockchain layer-2 techniques such as the payment channel network (PCN) \cite{yang2023optimal}.

\subsection{Threat Model} \label{Threat}
We follow Microsoft's STRIDE model \cite{shostack2008experiences} to comprehensively analyze the potential threats of semi-trusted intermediaries to \textsf{HiCoCS}. (i) \textit{Spoofing.} An attacker spoofs as an intermediary to perform \textsf{CSTx} processing to steal transaction data for additional profit. (ii) \textit{Tampering.} An intermediary may maliciously modify transaction data. (iii) \textit{Repudiation.} An intermediary denies specific operations performed in \textsf{CSTx}s, such as denying receipt of transferred assets from its initiator. (iv) \textit{Information disclosure.} An intermediary may expose the private data of the participants, such as transaction flows and asset holdings. (v) \textit{Denial of service.} An intermediary may actively or passively make \textsf{CSTx} services unavailable \cite{hayat2022ml}. (vi) \textit{Elevation of privilege.} An intermediary with limited privileges impersonates an intermediary with privileges to gain the ability to process \textsf{CSTx}s.

We state the following \textbf{assumptions}: The Hyperledger Fabric's channels (equivalent to shards in this paper) are secure and reliable. There is at least one honest intermediary in the intermediary group that joins the source and target channels to provide cross-shard service. The intermediary is usually well-funded to cope with concurrent \textsf{CSTx}s during a period (for potential liquidity issues, we discuss introducing a remedial mechanism to ensure the robustness of the system in Section \ref{Sync}). The public-private key distribution of the encryption and decryption functions is secure and does not involve key transmission security issues and quantum threats.

\subsection{System Goals} \label{Goal}
To meet the performance and privacy requirements for highly concurrent cross-shard transaction processing, HiCoCS needs to achieve the following goals.
\nopagebreak[4]
\begin{itemize}[leftmargin=*]
	\item[$\bullet$] \textbf{Robustness.} It allows large-scale \textsf{CSTx}s to be sent to the system simultaneously and ensures that \textsf{CSTx}s are properly received and processed with a high transaction success ratio.
	
	\item[$\bullet$] \textbf{Efficiency.} The system can process \textsf{CSTx}s efficiently and cost-effectively, i.e., it is characterized by high throughput, low confirmation latency, and low resource overhead.
	
	\item[$\bullet$] \textbf{Security.} The security of \textsf{CSTx}s includes (i) \textit{Data confidentiality.} The system ensures that the semi-trusted intermediaries cannot extract the transaction amount privacy. (ii) \textit{Transaction atomicity.} The system needs to satisfy the eventual atomicity of \textsf{CSTx}s. (iii) \textit{Service availability.} The system is able to provide continuous and effective service.
\end{itemize}

We emphasize that the \textit{focus} of HiCoCS is on achieving \textit{robustness} and \textit{efficiency}, and that privacy preservation in security is an optional feature provided. In this paper, no improvements are made to fully homomorphic cryptography (e.g., CKKS), as this is another work of independent interest, which we leave for future work.

%% file: 4-Building.tex
\section{System Design of HiCoCS} \label{Building}
\subsection{Design Outline}
Fig. \ref{4} is the design outline of HiCoCS, which shows the system's building blocks, mainly including the \textit{transaction generation module}, \textit{pre-processing module}, \textit{incremental accumulation module}, \textit{transaction synchronization module}, and \textit{transaction completion module}. In addition, the \textit{composite key reuse mechanism} is crucial to reduce the resource overhead of HiCoCS. The \textit{ciphertext conversion and computation mechanisms} are the key to achieving privacy protection. Their main functionalities are summarized below.

\begin{figure}[t]
	\centering
	\includegraphics[width=3.3in]{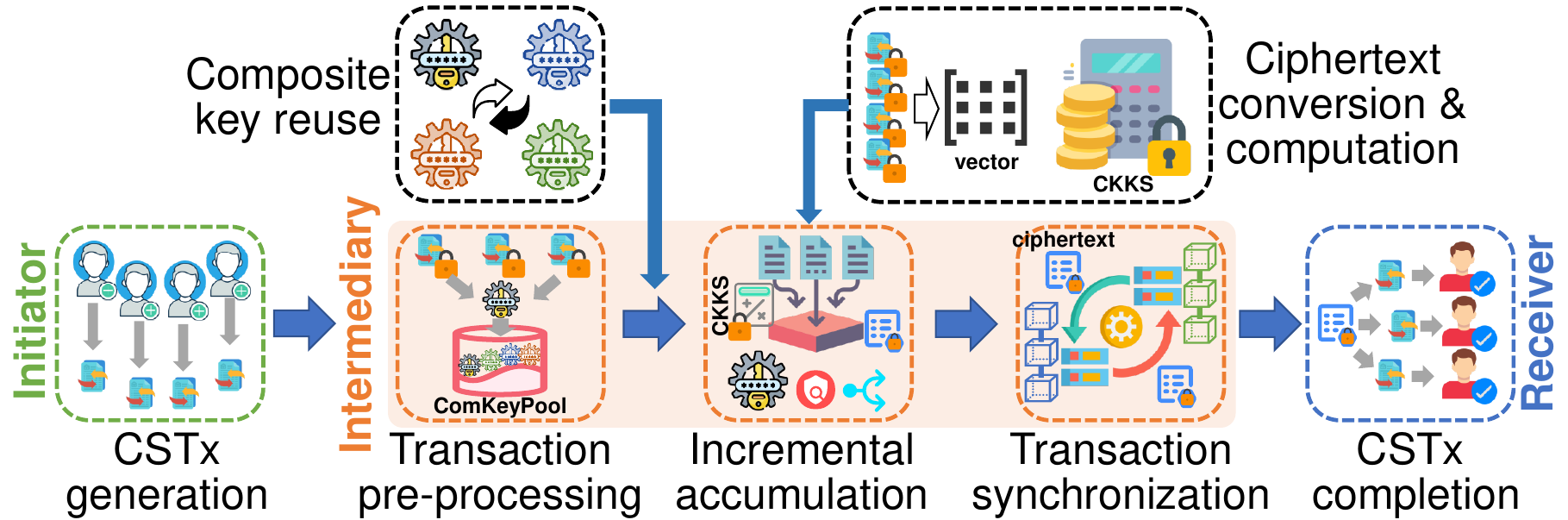}
	\vspace{-0.2cm}
	\caption{Building blocks of HiCoCS.}
	\label{4}
	\vspace{-0.5cm}
\end{figure}

\begin{itemize}[leftmargin=*]
	\item[$\bullet$] \textbf{Transaction generation module} is mainly responsible for packaging transactions. The packaged content includes the addresses of the transaction initiator and receiver, and the transaction inclusion (we omit its details).
	
	\item[$\bullet$] \textbf{Pre-processing module} mainly consists of collecting \textsf{CSTx} ciphertexts, constructing composite keys, and maintaining a composite key transaction pool \verb|ComKeyPool|.
	
	\item[$\bullet$] \textbf{Incremental accumulation module} adopts the idea of batch processing to summarize the intermediary-related \textsf{CSTx}s in the \verb|ComKeyPool|, including the operations of fuzzy query composite key, splitting composite key, and accumulating the final amount to be received under ciphertext.
	
	\item[$\bullet$] \textbf{Transaction synchronization module} requests the ciphertext processing result at the source channel from the incremental accumulation module via a cross-channel query, and synchronizes the pending \textsf{CSTx}s on the target channel with the intermediary's received \textsf{CSTx}s on the source channel.
	
	\item[$\bullet$] \textbf{Transaction completion module} distributes \textsf{CSTx}s to be transferred to accounts in the target shard (details omitted).
	
	\item[$\bullet$] \textbf{Composite key reuse mechanism} transforms multiple associated composite keys in \verb|ComKeyPool| into a single composite key, including composite key summarization and regeneration operations. It hands over the updated simplified \verb|ComKeyPool| to the incremental accumulation module.
	
	\item[$\bullet$] \textbf{Ciphertext conversion \& computation mechanisms} first convert the transaction ciphertext array into the CKKS ciphertext vectors for computation. Then, the ciphertext amounts are summarized (incrementally accumulated) to produce the ciphertext result, which is passed to the source channel for decryption and subsequent processing.
\end{itemize}

\subsection{Transaction Pre-Processing} \label{Pre}
Following the workflow, HiCoCS first encrypts the packaged \textsf{CSTx}s in the source channel to ensure the security of the \textsf{CSTx} transmission. Then, composite keys (virtual sub-brokers) are created for each intermediary under the \textsf{CSTx} ciphertext.

\textbf{Transmission pre-processing.} We use the advanced encryption standard (AES) algorithm \cite{selent2010advanced} to pre-process the \textsf{CSTx} messages. This serves two purposes: (i) The transaction initiator performs AES encryption of its initiated \textsf{CSTx} messages in the source channel, which ensures the security of \textsf{CSTx} message transmission. Of course, this requires ensuring the secure storage of keys. We utilize Fabric's private data collection to store the sender's AES key, and control that only authorized smart contract functions to access and use it for decryption. (ii) HiCoCS uses AES to convert packaged transaction messages (which may have inconsistent data lengths) into a uniformly formatted ciphertext (128-bit), which serves as the input to the CKKS algorithm used in the subsequent incremental accumulation module. This ensures a lightweight, encrypted transmission of \textsf{CSTx}s to the sharding network.

First, an initiator generates an AES symmetric key \verb|skey| for the encryption and decryption of each \textsf{CSTx} in the source channel. Then, for the transaction amount $\textsf{Amount}_i$, the initiator uses \verb|skey| to encrypt and obtain the \textsf{CSTx} ciphertext, i.e., $\textsf{AmountStrC}_i = \verb|EncryptAES(skey|, \textsf{Amount}_i\verb|)|$.

\begin{figure}[t]
	\centering
	\includegraphics[width=3.3in]{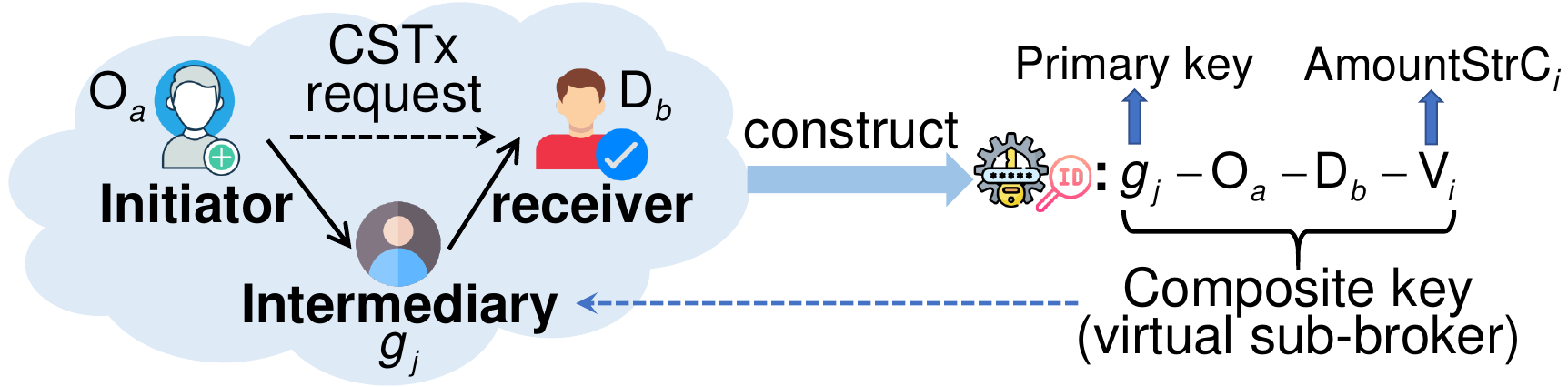}
	\vspace{-0.3cm}
	\caption{Composite key construction.}
	\label{4-2}
	\vspace{-0.5cm}
\end{figure}

\textbf{Composite key pre-processing.} To avoid concurrent \textsf{CSTx}s creating access conflicts to a single intermediary account/key, HiCoCS uses composite keys to combine multiple fields related to each \textsf{CSTx}. It constructs different composite keys as virtual sub-brokers of an intermediary to cache \textsf{CSTx}s, and adds each \textsf{CSTx} to a composite key transaction pool \verb|ComKeyPool|. Thus, it transforms the client's direct modification of the intermediary's state into the creation of composite keys, avoiding frequent access to a single key.

The construction of a composite key is shown schematically in Fig. \ref{4-2}. The intermediary group $\mathcal{G}$ contains $|\mathcal{G}|$ intermediaries, where each intermediary has an address/key of $g_j$. $\textsf{O}_a$ and $\textsf{D}_b$ denote the address/key of user $a$ on the source (original) channel and user $b$ on the target (destination) channel, respectively. Thus, a \textsf{CSTx} is denoted as $\textsf{CSTx}_i(\textsf{O}_a, \textsf{D}_b, g_j, \textsf{Amount}_i)$. The \textsf{CSTx} ciphertext message sent by the client via gRPC is $\textsf{CSTx}'_i(\textsf{O}_a, \textsf{D}_b, g_j, \textsf{AmountStrC}_i)$. The \textsf{CSTx} amount (value) ciphertext $\textsf{AmountStrC}_i$ is simplified to $\textsf{V}_i$. In addition, to ensure the uniqueness of the composite key, $\textsf{V}_i $ contains the timestamp of the transaction (for the case of multiple concurrent \textsf{CSTx}s with the same initiator, receiver, intermediary, and transaction amount). Thus, the pre-processing module constructs a composite key (virtual sub-broker) of $\textsf{CSTx}'_i$ as $\textsf{CK}_i: g_j - \textsf{O}_a - \textsf{D}_b - \textsf{V}_i$, and $g_j$ is used as the primary key of this composite key. The intermediary then adds $\textsf{CK}_i$ to the \verb|ComKeyPool| $\leftarrow \{ \textsf{CK}_1, \textsf{CK}_2, ..., \textsf{CK}_{i-1} \}$. When the time reaches the periodic settlement time threshold, i.e., $t = T_{\textsf{settle}}$, the intermediary group settles $i$ \textsf{CSTx}s corresponding to \verb|ComKeyPool| $\leftarrow \{ \textsf{CK}_1, \textsf{CK}_2, ..., \textsf{CK}_{i} \}$.

\textbf{Discussion.} We explore a potential double-spending attack and its corresponding countermeasure. Suppose an attacker initiates two \textsf{CSTx}s simultaneously, with a combined total amount exceeding its account balance, and these two \textsf{CSTx}s are processed by different intermediaries and packaged into separate transaction batches. During validation, only one of the \textsf{CSTx}s can succeed, while the other fails due to insufficient balance, leading to the rejection of the entire batch and negatively impacting other valid \textsf{CSTx}s. HiCoCS employs timestamps to ensure that even if two \textsf{CSTx}s are initiated simultaneously, they will have distinct timestamps corresponding to the composite key, causing the latter transaction in the queue to fail due to insufficient balance. Consequently, HiCoCS prevents this attack at its source. Even if an attack occurs, since the composite key contains the amount of transactions that failed due to malicious initiation, HiCoCS rolls back only the failed \textsf{CSTx} and a small number of dependent \textsf{CSTx}s, rather than the entire batch. This ensures that other valid \textsf{CSTx}s are completed successfully.

\subsection{Incremental Accumulation} \label{Incremental}
The process of transaction incremental accumulation under ciphertext is shown in Alg. \ref{alg-cstxAcc}.

\begin{algorithm}[h]
	{\footnotesize
		\DontPrintSemicolon
		\normalem 
		\caption{Incremental Accumulation Algorithm} \label{alg-cstxAcc}
			\KwInput{The intermediary's key $g_j$, a plaintext scaling factor $\Delta$, an exchange rate $\textsf{C}_\textsf{rate}$}
			\KwOutput{The final ciphertext summation result $\textsf{C}_\textsf{finalSum}$}
			\SetKwFunction{FcstxAcc}{\textsf{cstxAcc}}
			\SetKwProg{Fn}{Function}{:}{}
			\Fn{\FcstxAcc{$\cdot$}}{				
				\If{$t = T_{\textsf{settle}}$}{ \label{alg-cstxAcc-2}
					\tcp{Ciphertext collection}
					\tcp{The $\_$ symbol below denotes the prefix}
					\mbox{$\texttt{ComKeyPool}_j \leftarrow \texttt{GetStateByPartialCompositeKey}\texttt{(}\_, g_j\texttt{)}$} \\ \label{alg-cstxAcc-3}
					\For{$i = 0$ \KwTo $\lvert \textnormal{\texttt{ComKeyPool}}_j \rvert - 1$}{
						$\texttt{compositeKeyObject}_i \leftarrow \texttt{ComKeyPool}_{j,i}$ \\
						\tcp{Get compositeKey}
						$\textsf{CK}_{i} \leftarrow \texttt{compositeKeyObject}_i.\texttt{getKey()}$ \\
						\tcp{Split compositeKey}
						$\_, \textsf{D}_{i}, \textsf{V}_{i} \leftarrow \texttt{SplitCompositeKey(} \textsf{CK}_{i} \texttt{)}$ \\
						Add $\{\textsf{D}_{i}, \textsf{V}_{i}\}$ to the pending transaction set \\
						$\textsf{aesStrCipher}.\texttt{add(}\textsf{V}_{i}\texttt{)}$ \\ \label{alg-cstxAcc-9}
					}
					\tcp{Ciphertext conversion}
					\For{$i = 0$ \KwTo $\lvert \textnormal{\textsf{aesStrCipher}} \rvert - 1$}{ \label{alg-cstxAcc-10}
						\tcp{Call the \texttt{convert()} interface}
						$\textsf{Amount}_i \leftarrow \texttt{DecryptAES(skey}, \textsf{V}_i\texttt{)}$ \\ \label{alg-cstxAcc-11}
						$\textsf{Amounts}[i] \leftarrow \texttt{complex(}\textsf{Amount}_i, 0\texttt{)}$ \\\label{alg-cstxAcc-12}
						$\textsf{Amounts}.\texttt{add(}\textsf{Amounts}[i]\texttt{)}$ \\
					}
					$\textsf{m}(\textsf{X}) \leftarrow  \texttt{Encode(}\textsf{Amounts}, \Delta\texttt{)}$ \\
					$\texttt{(pk, sk)} \leftarrow \texttt{NewKeyPair()}$ \\
					$\textsf{CKKSCipher} \leftarrow \texttt{EncryptCKKS(pk}, \textsf{m}(\textsf{X})\texttt{)}$ \\ \label{alg-cstxAcc-16}
					\tcp{Ciphertext computation}
					$\texttt{rlk} \leftarrow \texttt{GenRelinearizationKey()}$ \\ \label{alg-cstxAcc-17}
					$\texttt{gk} \leftarrow \texttt{GenGaloisKey()}$ \\ \label{alg-cstxAcc-18}
					$\texttt{evaluator} \leftarrow \texttt{NewEvaluator()}$ \\ \label{alg-cstxAcc-19}
					$\texttt{RotationKey} \leftarrow \texttt{GenRotationKeys(gk, sk)}$ \\ \label{alg-cstxAcc-20}
					$\texttt{eval} \leftarrow \texttt{evaluator.WithKey(RotationKey, rlk)}$ \\ \label{alg-cstxAcc-21}
					\For{$i = 0$ \KwTo $\lceil \textnormal{\texttt{batch}}/n \rceil - 1$}{ \label{alg-cstxAcc-22}
						$\texttt{eval.InnerSum(}\textsf{CKKSCipher}, \texttt{batch}, n, \textsf{CKKSCipher}\texttt{)}$ \\ \label{alg-cstxAcc-23}
					}
					$\textsf{C}_\textsf{sum} \leftarrow \textsf{CKKSCipher}[0]$ \\ \label{alg-cstxAcc-24}
					\tcp{Asset exchange}
					$\textsf{C}_\textsf{finalSum} \leftarrow \texttt{evaluator.MulNew(}\textsf{C}_\textsf{sum}, \textsf{C}_\textsf{rate}\texttt{)}$ \\ \label{alg-cstxAcc-25}
				}
				\textbf{return} $\textsf{C}_\textsf{finalSum}$
			}	
		}
\end{algorithm}

\textbf{Amount ciphertext collection.} (Lines \ref{alg-cstxAcc-2}-\ref{alg-cstxAcc-9}) If the current time $t = T_{\textsf{settle}}$, the intermediary starts to settle the $\textsf{CSTx}$s accumulated in the last period. First, a fuzzy query is performed on the \verb|ComKeyPool| based on each intermediary's key $g_j$ (primary key). As a result, the system obtains a $\textsf{CSTx}$ set associated with the intermediary $g_j$, i.e., $\verb|ComKeyPool|_{j} \leftarrow \{ \textsf{CK}_1, \textsf{CK}_2, ..., \textsf{CK}_{i} \}$. Then, the intermediary traverses $\verb|ComKeyPool|_{j}$ and calls Fabric's composite key splitting function $\verb|SplitCompositeKey()|$ to obtain the set of \textsf{CSTx} amount ciphertexts $\textsf{aesStrCipher}=\{ \textsf{V}_1, \textsf{V}_2, ..., \textsf{V}_n \}$. Also, it adds $\{\textsf{D}_{i}, \textsf{V}_{i}\}$ to the pending transaction set.

\textbf{Ciphertext conversion.} (Lines \ref{alg-cstxAcc-10}-\ref{alg-cstxAcc-16}) Next, the system performs a ciphertext conversion on $\textsf{aesStrCipher}$ for CKKS homomorphic privacy calculation of the accumulated $\textsf{CSTx}$s’ amount, as shown in Fig. \ref{4-3}. First, the intermediary traverses \textsf{aesStrCipher} and calls the AES decryption function to decrypt each $\textsf{V}_i$ to get the amount of each \textsf{CSTx}, i.e., line \ref{alg-cstxAcc-11}. Note that \verb|skey| is not leaked here because \verb|skey| is stored within the initiator's private data collection, and the intermediary calls a ciphertext conversion interface \verb|convert()| to get a CKKS ciphertext vector, i.e., $\textsf{Amount}_i$ is a computed intermediate value, which is not available to the intermediary. Since CKKS is represented in complex space, if an operation is performed on a real number, the imaginary part of the complex number needs to be set to zero first, i.e., line \ref{alg-cstxAcc-12}. Then, the array ($\textsf{Amounts}$) is encoded as a CKKS plaintext $\textsf{m}(\textsf{X})$ in integer form, i.e., $\textsf{m}(\textsf{X}) = \textsf{Amounts} \cdot \Delta$, where $\Delta$ is a plaintext scaling factor and $\Delta > 0$. Then, the system generates a public-private key pair $(\verb|pk|, \verb|sk|)$ (\verb|sk| in the privacy data collection) required for CKKS in the source channel. It encrypts the plaintext $\textsf{m}(\textsf{X})$ into ciphertext, i.e., line \ref{alg-cstxAcc-16}. It satisfies: (i) $\textsf{CKKSCipher}=(\textsf{c}_0+\textsf{m}, \textsf{c}_1) \in R^2_\textsf{Q}$, where $(\textsf{c}_0, \textsf{c}_1)$ is the randomized instance of ring learning with errors (RLWE).  \textsf{Q} is the maximum ciphertext modulus. Also, (ii) $\textsf{c}_0 + \textsf{c}_1 \cdot \verb|sk| = \textsf{e}$, where \textsf{e} is the noise.

\begin{figure}[t]
	\centering
	\includegraphics[width=3.3in]{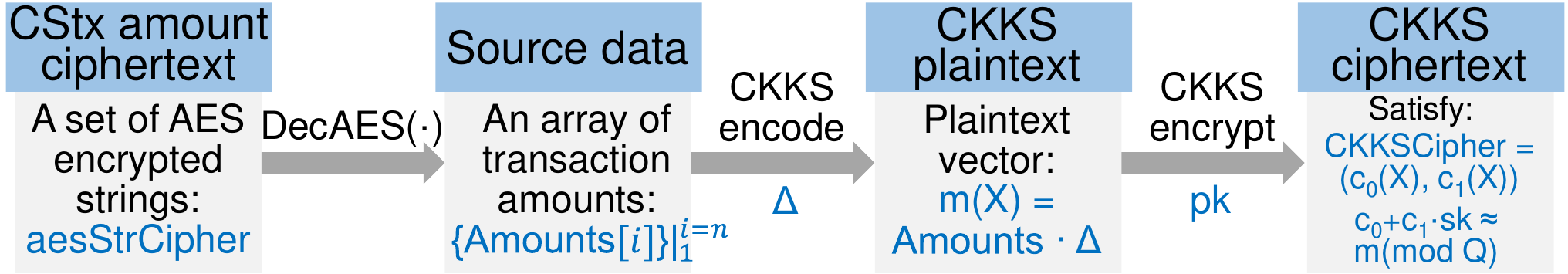}
	\vspace{-0.2cm}
	\caption{Ciphertext conversion process.}
	\label{4-3}
	\vspace{-0.3cm}
\end{figure}

\textbf{Ciphertext computation.} (Lines \ref{alg-cstxAcc-17}-\ref{alg-cstxAcc-24}) After obtaining \textsf{CKKSCipher}, the system computes the accumulated amount of the \textsf{CSTx}s in ciphertext, i.e., sums up the components of the ciphertext vector ($c_i$). We realize the incremental accumulation of the \textsf{CSTx} amount through the ciphertext rotation operation in the following form:

{\small
\begin{equation}
	\begin{aligned}
		\textsf{E}(c_0, c_1, ..., c_{n-1}) \xrightarrow{\rm{Rotate} \  \textit{k} \ vectors} \textsf{E}(c_k, ..., c_{n-1}, c_0, ..., c_{k-1}).
	\end{aligned}
	\nonumber 
\end{equation}
}

Let the number of components in the ciphertext vector be $n$ and the ciphertext group size be \verb|batch|. First, initialize the parameters, including a relinearization key \verb|rlk| (line \ref{alg-cstxAcc-17}) and a Galois key \verb|gk| (line \ref{alg-cstxAcc-18}). Then, a ciphertext evaluator is created, i.e., line \ref{alg-cstxAcc-19}. The private key \verb|sk| and the Galois key \verb|gk| are utilized to generate a key for internal rotation \verb|RotationKey| (line \ref{alg-cstxAcc-20}). Then, create a shallow copy of the \verb|evaluator|, i.e., line \ref{alg-cstxAcc-21}. \verb|eval| is a new \verb|evaluationKey| and shares the temporary buffer with the previous one. The \verb|batch| components of the ciphertext vector are then rotated $\lceil \verb|batch|/n \rceil$ times in groups of $n$, each time evaluating the group's inner sum (calling \verb|evaluator.InnerSum()|), and finally, the groups are summed (lines \ref{alg-cstxAcc-22}-\ref{alg-cstxAcc-23}). The final value of all ciphertext slots is an inner sum, and we take the first slot returned as the sum of the transaction ciphertexts, i.e., $\textsf{C}_\textsf{sum}$ (line \ref{alg-cstxAcc-24}).

We emphasize that the above ciphertext conversion and computation are performed in the source channel. Thus, no transaction amount privacy is leaked by intermediaries.

\textbf{Confidential asset exchange.} (Line \ref{alg-cstxAcc-25}) Finally, since the asset values of the two shards may be different, i.e., there exists an exchange rate $\textsf{C}_\textsf{rate}$. To prevent an attacker from deriving the value of the two assets through the exchange rate, it is also necessary to encrypt the exchange rate into a ciphertext. Thus, we utilize ciphertext multiplication to achieve asset cross-shard exchange, i.e., the final ciphertext summation result is $\textsf{C}_\textsf{finalSum} = \verb|evaluator.MulNew(|\textsf{C}_\textsf{sum}, \textsf{C}_\textsf{rate}\verb|)|$.

\textbf{Discussion.} HiCoCS uses homomorphic encryption techniques to preserve the privacy of \textsf{CSTx}s. It focuses only on the privacy of the data and not on the identity of the participants. If participants are concerned about the privacy of their identities, they may consider obfuscating the addresses before the transactions \cite{wu2021towards}. As it is not the focus of this paper, we only briefly describe this approach here. When users need to protect their identity privacy, HiCoCS provides an optional mixing service for identity obfuscation. This service can be provided by an organization consisting of multiple intermediaries, who split each transaction served by each of them into multiple sub-transactions to be mixed within the organization, and ultimately, multiple intermediaries transfer funds to the receiver in multiple passes. The level of transaction obfuscation is correlated with an additional fee paid by the user. In any case, compared to the vanilla version, our approach still has the advantage in guaranteeing the confidentiality of the original transaction data. The intermediaries perform the cumulative transaction amount computation in the ciphertext space, which can effectively mitigate the threats stated in \S \ref{Threat}.

\subsection{Transaction Synchronization and Completion} \label{Sync}
After obtaining the CKKS ciphertext processing results, the source and target channels/shards synchronize that cipher state results, including the original and exchanged results, i.e.,  $\textsf{C}_\textsf{sum}$ and $\textsf{C}_\textsf{finalSum}$. Then, they call the \verb|DecryptCKKS()| function to decrypt the two ciphertext results and decode them in plaintext space, respectively. The initiator of the source channel gets the decryption results and records them in its ledger. The plaintext code for the amount to be transferred to the receiver of the target channel is as follows

{\small
\begin{equation}
	\begin{aligned}
		\textsf{inAmount} = \verb|encoder.Decode(|\verb|DecryptCKKS(|\verb|sk|, \textsf{C}_\textsf{finalSum}\verb|))|.
	\end{aligned}
	\nonumber 
\end{equation}
}

The plaintext code for the amount ultimately to be deducted by the intermediary is as follows

{\small
\begin{equation}
	\begin{aligned}
		\textsf{outAmount} = \verb|encoder.Decode(|\verb|DecryptCKKS(|\verb|sk|, \textsf{C}_\textsf{sum}\verb|))|.
	\end{aligned}
	\nonumber 
\end{equation}
}

Finally, since the decoded plaintext is in the form of a complex array, we only need to obtain its real part. The system updates the ledger states of the receivers and intermediaries based on the pending transaction set. Eventually, the \textsf{CSTx}s are completed.

\textbf{Discussion.} If affected by liquidity issues, the intermediary accounts may not have sufficient funds to complete the current period of \textsf{CSTx}s. To enhance the system's robustness, we propose introducing a remedial mechanism to address this issue. We adopt the concepts of real-time monitoring via smart contracts and the joint maintenance of liquidity pools \cite{zheng2018detailed, jeong2023efficient} to ensure the stability and continuity of cross-shard fund transfers under suboptimal conditions in HiCoCS. We briefly describe the multidimensional components of the mechanism: (i) \textit{Dynamic fund management.} By running real-time monitoring algorithms on smart contracts, the system adjusts the intermediary's pool configuration to replenish funds automatically. (ii) \textit{Distributed liquidity pool.} Establish a liquidity pool jointly maintained by multiple intermediaries to share the funding pressure caused by highly concurrent transactions. (iii) \textit{Prioritized processing mechanism.} In the event of fund constraints, prioritize the processing of small amount transactions to improve the efficiency of fund usage.

\begin{figure}[t]
	\centering
	\includegraphics[width=3.1in]{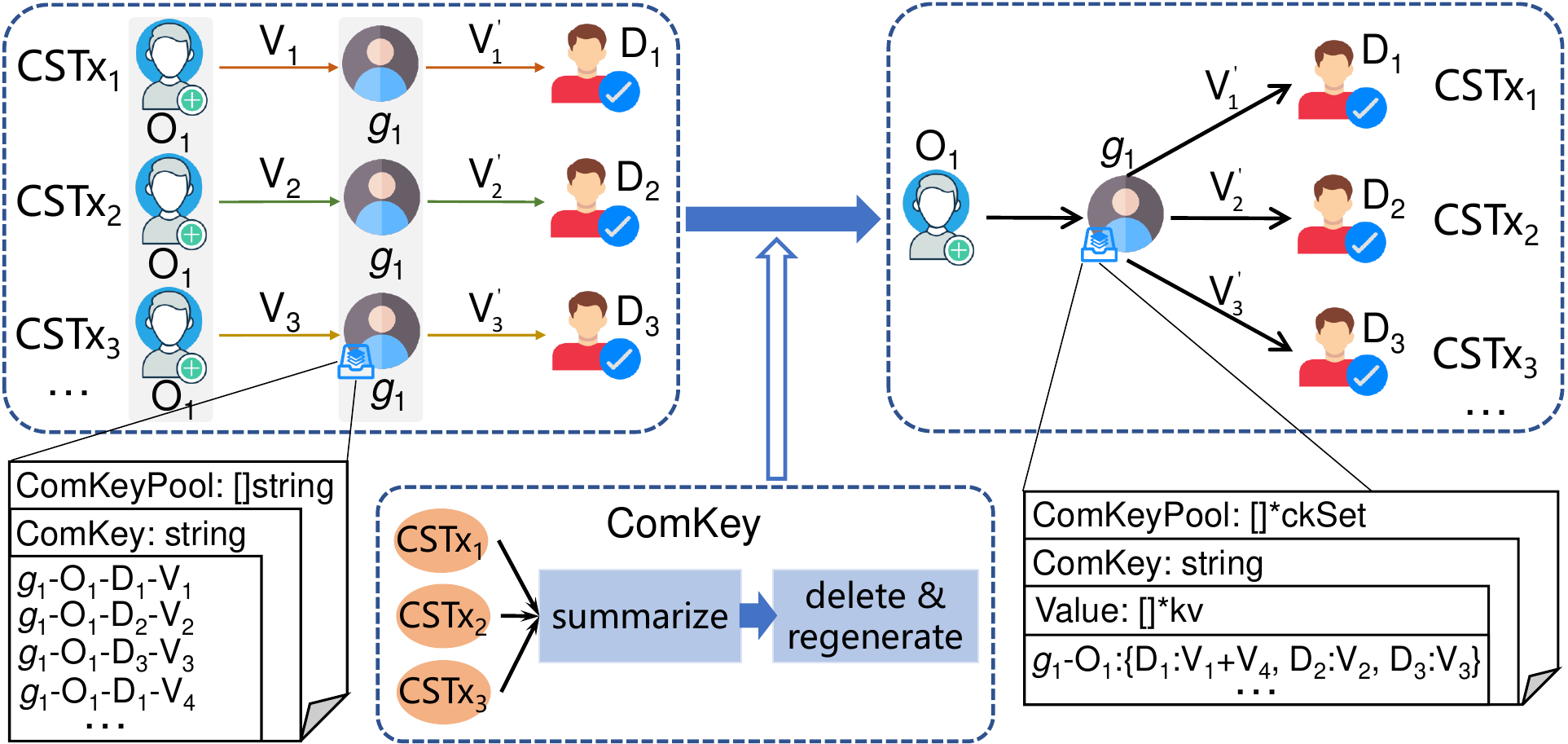}
	\vspace{-0.2cm}
	\caption{Composite key reuse.}
	\label{4-4}
	\vspace{-0.3cm}
\end{figure}

\subsection{Composite Key Reuse} \label{Reuse}
We adopt an idea of Proof-of-Equivalence (PoE) \cite{chatzopoulos2020mneme}. It periodically summarizes data and generates equivalent blocks that require less storage space, which significantly saves resources. We implement a periodically running Composite Key Proof of Equivalence (CKPoE) protocol in the composite key reuse mechanism. It consists of composite key summarization and regeneration phases.

A randomly selected intermediary $g_j \in \mathcal{G}$ in the intermediary group $\mathcal{G}$ stores the new set of composite keys, combines the composite keys generated from \textsf{CSTx}s in which the participants are (partially) the same, and then regenerates their equivalent composite keys. For example, in Fig. \ref{4-4}, assume that the composite key transaction pool (\verb|ComKeyPool|) (left) has four composite keys. $\textsf{CK}_1: g_1 - \textsf{O}_1 - \textsf{D}_1 - \textsf{V}_1$, $\textsf{CK}_2: g_1 - \textsf{O}_1 - \textsf{D}_2 - \textsf{V}_2$, $\textsf{CK}_3: g_1 - \textsf{O}_1 - \textsf{D}_3 - \textsf{V}_3$, and $\textsf{CK}_4: g_1 - \textsf{O}_1 - \textsf{D}_1 - \textsf{V}_4$. Actually, the initiator $\textsf{O}_1$ transfers a total amount of $\textsf{V}_1 + \textsf{V}_2 + \textsf{V}_3 + \textsf{V}_4$ to the intermediary $g_j$, where a total of $\textsf{V}_1+\textsf{V}_4$ needs to be transferred to the receiver $\textsf{D}_1$. Thus, we create a new \verb|ComKeyPool| (right) of pending $\textsf{CSTx}$s, which deletes the original 4 composite keys and adds a new composite key $g_1 - \textsf{O}_1:\{ \textsf{D}_1: \textsf{V}_1+\textsf{V}_4, \textsf{D}_2: \textsf{V}_2, \textsf{D}_3: \textsf{V}_3 \}$.

\begin{algorithm}[h]
	{\footnotesize
		\DontPrintSemicolon
		\normalem 
		\caption{Composite Key Reuse Algorithm} \label{alg-ckReuse}
			\KwInput{The current composite key transaction pool \texttt{ComKeyPool}, the intermediary group $\mathcal{G}$}
			\KwOutput{The equivalent composite key transaction pool after reusing updates \texttt{ComKeyPool}}
			\SetKwFunction{FckReuse}{\textsf{ckReuse}}
			\SetKwProg{Fn}{Function}{:}{}
			\Fn{\FckReuse{$\cdot$}}{
				\For{$j = 0$ \KwTo $\lvert \mathcal{G} \rvert - 1$}{ \label{alg-ckReuse-2}
					\tcp{Composite key summarization}
					\tcp{Query the pending transaction set}
					\mbox{$\texttt{CKSet} \leftarrow \texttt{queryCompositeKeyByPartial(}g_j, \texttt{ComKeyPool)}$} \\ \label{alg-ckReuse-3}
					$\texttt{ComKeyMap} = \emptyset$ \\
					\For{$i = 0$ \KwTo $\lvert \textnormal{\texttt{CKSet}} \rvert - 1$}{
						$\_, \textsf{O}_{i}, \textsf{D}_{i}, \textsf{V}_{i} \leftarrow \texttt{SplitCompositeKey(} \textsf{CK}_{i} \texttt{)}$ \\
						$\textsf{CKNew}_{i} \leftarrow g_j \sim \textsf{O}_{i}$ \tcp*{\mbox{$\sim:$ A concatenator}}  \label{alg-ckReuse-7}
						\If{\textnormal{$\texttt{IsExist(}\textsf{D}_{i}, \texttt{ComKeyMap}\texttt{)}$}}{ 
							$\texttt{ComKeyMap}[\textsf{D}_{i}] += \textsf{V}_{i}$ \\
						}
						\Else{
							$\texttt{ComKeyMap}[\textsf{D}_{i}] = \textsf{V}_{i}$ \\
						}\label{alg-ckReuse-11}
						\tcp{Construct new composite keys}
						$\texttt{TempCK}_j \leftarrow \texttt{createCKNew(}\textsf{CKNew}_{i}, \texttt{ComKeyMap)}$ \\ \label{alg-ckReuse-12}
						$\texttt{TempComKeyPool.add(TempCK}_j\texttt{)}$ \\ \label{alg-ckReuse-13}
					}
					\tcp{Composite key regeneration}
					\For{$i = 0$ \KwTo $\lvert \textnormal{\texttt{ComKeyPool}} \rvert - 1$}{ \label{alg-ckReuse-14}
						$\texttt{delete(}\textsf{CK}_{i}\texttt{)}$ \\
					}
					$\texttt{ComKeyPool} \leftarrow \texttt{TempComKeyPool}$ \\
					$\texttt{destory(TempComKeyPool)}$ \\ \label{alg-ckReuse-16}
				}
				
				\textbf{return} $\texttt{ComKeyPool}$
			}
			
		}
\end{algorithm}

We briefly describe the CKPoE protocol process in Alg. \ref{alg-ckReuse}.

\textbf{Composite key summarization.} (Lines \ref{alg-ckReuse-2}-\ref{alg-ckReuse-13}) The periodically elected intermediary $g_j$ listens to the composite keys added to \verb|ComKeyPool|, checks all the accounts involved, and summarizes the composite keys with the same participants and their amounts. The idea is to count the final execution results of multiple transactions at once and then record the equivalent transaction results. First, the intermediary group $\mathcal{G}$ is traversed to obtain the key of each intermediary. Then, the pending transaction set (\verb|CKSet|) associated with each intermediary is queried based on its key. A map set (\verb|ComKeyMap|) is constructed to store the new composite keys. Next, the \verb|CKSet| is traversed, and each composite key is split to obtain its original account's key, the destination account's key, and the transaction amount. Then (lines \ref{alg-ckReuse-7}-\ref{alg-ckReuse-11}), the original account's key and the intermediary's key are combined as a new composite key. Meanwhile, the total amount of the destination account is counted under the ciphertext and used as the value of the new composite key. Finally (lines \ref{alg-ckReuse-12}-\ref{alg-ckReuse-13}), the new composite key is constructed to append a temporary composite key transaction pool (\verb|TempComKeyPool|).

\textbf{Composite key regeneration.} (Lines \ref{alg-ckReuse-14}-\ref{alg-ckReuse-16}) Intermediary $g_j$ updates \verb|ComKeyPool| based on the result of the summarization phase, i.e., \verb|TempComKeyPool|. After validation by other intermediaries, the new \verb|ComKeyPool| is broadcast. Each user who receives it can delete the previously outdated \verb|ComKeyPool| and update it to the latest equivalent.

{
	\setcounter{algocf}{0} 
	\SetAlgorithmName{{\small \texttt{Solidity code}}}{local}{List of Local Algorithms}
	\begin{algorithm}[h]
		\DontPrintSemicolon
		\normalem 
		\caption{{\small Composite Key Implementation Example}} \label{alg1}
		{\footnotesize
			\tcp{CSTx structure}
			\verb|struct| CSTx \{ \\ \label{l1}
			\mbox{} \qquad \verb|address| intermediary; \\
			\mbox{} \qquad \verb|address| initiator; \\
			\mbox{} \qquad \verb|address| receiver; \\
			\mbox{} \qquad \verb|string| amountC; \tcp*{Amount ciphertext}
			\} \\ \label{l6}
			\tcp{Data mapping}
			\verb|mapping(string| $\Rightarrow$ CSTx\verb|)| \verb|private| cstxs; \\ \label{l7}
			\tcp{Composite key creation}
			\verb|function| \textsf{createCompositeKey}\verb|(string memory| objectType, \verb|address| intermediary, \verb|address| initiator, \verb|address| receiver, \verb|string memory| amountC\verb|)| \verb|public| \{ \\ \label{l8}
			\mbox{} \qquad \verb|string memory| CK = \verb|abi.encodePacked(|objectType, ``$|$", intermediary, ``$|$", initiator, ``$|$", receiver, ``$|$", amountC\verb|)|; \\\label{l10}
			\mbox{} \qquad cstxs[CK] = CSTx\verb|(|intermediary, initiator, receiver, amountC\verb|)|; \\
			\} \\ \label{l11}
			\tcp{Composite key query}
			\verb|function| \textsf{queryCompositeKey}\verb|(string memory| CK\verb|)| \verb|public view returns| \verb|(|CSTx \verb|memory|\verb|)| \{ \\  \label{l12}
			\mbox{} \qquad  \verb|require(|cstxs[CK].initiator != \verb|address(|0\verb|)|, ``\texttt{CSTx does not exist.}"\verb|)|; \\
			\mbox{} \qquad \verb|return| cstxs[CK]; \\
			\} \\ \label{l15}
		}
	\end{algorithm}
}

\subsection{Ideas for Extension to Other Blockchains} \label{App-A3}
Further, we discuss ideas for applying HiCoCS to other blockchains. While our main focus is on solving the problem in Hyperledger Fabric, it proposes a way to mitigate the problem of cross-shard transaction conflicts in permissioned blockchain systems. We also consider that similar techniques and approaches can be applied in other blockchain systems, and thus, HiCoCS has a certain degree of generalizability.

Rich Web3 applications need to be served by different types of blockchains. Our vision is to design HiCoCS as a generalized cross-sharding middleware. Thus, HiCoCS needs to provide an easily scalable paradigm. Though only Hyperledger Fabric's chaincode currently provides the APIs for composite keys, other blockchains can implement similar functionality. As long as the blockchain smart contract can implement the composite key, then HiCoCS can be deployed on it. For example, EEA \cite{eea} and ChainMaker \cite{chainmaker} are the mainstream enterprise-grade permissionless and permissioned blockchains, respectively. They both support the \verb|Solidity| language for developing smart contracts. We provide a simplified example, as shown in \verb|Solidity code| \ref{alg1}. First (lines \ref{l1}-\ref{l6}), the data structure associated with the composite key and \textsf{CSTx} is defined. We use a \verb|mapping| to store the data related to the composite key (line \ref{l7}). Then, insert the concatenator ``$|$" between each of the two attributes to create the composite key (line \ref{l10}). Finally (lines \ref{l12}-\ref{l15}), a query composite key function is shown.

%% file: 5-Security.tex
\section{Security Analysis} \label{Security}
In this section, we briefly analyze the security of HiCoCS with respect to the stated threats in Section \ref{Threat} and the security goals in Section \ref{Goal}.

\subsection{Data Confidentiality}
\begin{thm} \label{thm1}
	In HiCoCS, for all cross-shard transactions, there is no probabilistic polynomial time (PPT) attacker $\mathcal{A}$ can spoof as an intermediary to successfully extract or infer transaction amount privacy information from the intercepted data.
\end{thm}
\begin{proof}
	HiCoCS uses the AES encryption algorithm to encrypt and pre-process the transmission of cross-shard transactions, and there is no unauthorized entity (including intermediaries) can decrypt the message without the key. Further, HiCoCS performs fully homomorphic encryption computation on the cross-shard transactions using the CKKS encryption algorithm. According to the security properties of the CKKS encryption algorithm \cite{cheon2017homomorphic}, the intermediary can only operate in the ciphertext space when processing transactions. Due to the use of Fabric's private data collection to manage the keys, even in the face of a passive attack \cite{li2021security}, attacker $\mathcal{A}$ cannot crack the private key by algebraic operations based on the hints and does not have the prerequisites to crack the private key. Thus, HiCoCS ensures that even in the existence of attacker $\mathcal{A}$, the semi-trusted intermediary cannot directly or indirectly obtain the plaintext transaction amount data.
\end{proof}

As HiCoCS provides data confidentiality, the threats of \textit{spoofing}, \textit{tampering}, \textit{information disclosure}, and \textit{elevation of privilege} are effectively defended against.

\subsection{Transaction Atomicity}
\begin{thm} \label{thm2}
	The cross-shard transactions in HiCoCS satisfy eventual atomicity, including the conditions that (i) if a cross-shard transaction is successfully executed, the states of the sharding ledgers involved in the transaction are consistent, and (ii) if a cross-shard transaction fails to execute, it must be rolled back on the involved shards.
\end{thm}
\begin{proof}
	HiCoCS uses a message-passing approach to process cross-shard transactions to guarantee atomicity. Specifically, HiCoCS collects transactions through a composite key transaction pool and then processes transactions in batches in the source and target channels. This means that the execution of transactions does not immediately satisfy atomicity. Only when the target channel listens for all sending transactions to be confirmed and executed and the final settlement performed is atomicity finally achieved. While this approach does not require immediate atomicity, it guarantees eventual atomicity \cite{wang2019monoxide, huang2022brokerchain}. Regardless of the circumstances, HiCoCS will eventually ensure that cross-shard transactions are either completed successfully and result in consistent state changes or are rolled back to prevent any partial or inconsistent outcomes. Due to the uniqueness of the composite key corresponding to a transaction, the system can precisely roll back invalid transactions and their dependencies. Thus, HiCoCS ensures that cross-shard transactions satisfy eventual atomicity.
\end{proof}

Existing cross-sharding schemes \cite{kokoris2018omniledger, al2018chainspace} achieve strong atomicity based on a two-phase commit protocol. It is unsuitable for high concurrency scenarios as a locking mechanism can seriously affect performance. Thus, HiCoCS employs the message-passing method to guarantee the eventual atomicity, mitigating the threat of \textit{repudiation}.

\subsection{Service availability}
\begin{thm} \label{thm3}
	HiCoCS ensures high service availability by deploying an intermediary group $\mathcal{G}$ to provide cross-shard transaction services to ensure that the system can still provide services in the face of disruptions or attacks.
\end{thm}
\begin{proof}
	In the HiCoCS architecture, multiple intermediary nodes work together to participate in the highly concurrent processing of transactions, and even if some nodes are attacked or fail, at least one node $g_j \in \mathcal{G}$ can still take over their work and keep the system running normally. This design effectively spreads the risk and reduces the possibility of a single point of failure, thus ensuring the service availability of the system.
\end{proof}

Thus, HiCoCS can effectively resist \textit{denial of service} threats and fulfill high availability requirements.

%% file: 6-PerformanceAnalysis.tex
\section{Performance Evaluation} \label{Evaluation}
\subsection{Settings} 
\textbf{Experimental prototype.} We use \verb|Golang| in Hyperledger Fabric v2.4\footnote{https://github.com/hyperledger/fabric/tree/release-2.4} to develop and implement a HiCoCS prototype. It leverages Hyperledger Fabric's multi-channel architecture to construct a scalable, multi-shard network comprising up to 128 nodes across 32 shards. We utilize Docker containers\footnote{https://github.com/jenkinsci/docker} as the execution environment for smart contracts. We are going to open-source the code\footnote{https://github.com/cwf1999/HiCoCS} following the paper's acceptance.

\textbf{Testbed.} We evaluate the performance of the prototype on a machine with an Intel i7-13700 CPU and 128GB of RAM. Fabric's \texttt{BatchTimeout} is set to 2 seconds. We adjust the block size setting from 10 MB to 160 MB. We evaluate the performance of HiCoCS with and without homomorphic encryption enabled for privacy preservation in concurrent cross-shard trading scenarios (with 1,000 concurrent threads). The transfer amounts for \textsf{CSTx}s are randomly selected from positive real numbers in blocks collected from the Ethereum blockchain (block height: 9,000,000-10,000,000). To model the system's transaction conflicts, we adopt both active and passive approaches: (i) \textit{Active.} Keeping the block size constant, we vary the skewness $f$ of the generated transactions, i.e., the proportion of \textsf{CSTx}s processed through the same intermediary account. (ii) \textit{Passive.} Fixing the skewness $f$, we vary the block size, and transactions are transferred among randomly selected participants. Changing the block size indirectly controls the probability of transaction conflicts within a block.

\subsection{Baselines and Metrics}
\textbf{Baselines.} We compare HiCoCS with three baselines: an implemented vanilla version, and two simulated state-of-the-art 2PL/OCC sharding schemes. (i) \textit{Vanilla version}. Since HiCoCS is the first scheme to propose high concurrency cross-sharding on Hyperledger Fabric, we compare its vanilla version as a baseline. The concept is introduced in Section \ref{intro}. (ii) \textit{AHL+} \cite{dang2019towards}. This scheme ensures the consistency and atomicity of \textsf{CSTx}s by locking resources and reaching agreements across committees through a 2PL protocol. (iii) \textit{Meepo} \cite{zheng2022meepo}. lt reduces transaction conflicts and enhances communication efficiency by introducing cross-epoch and cross-call protocols for ordered cross-shard communication (i.e., OCC) between blocks.

\textbf{Metrics.} We measure the performance of HiCoCS using the following metrics. (i) \textit{Transaction success rate (TSR).} The percentage of successfully executed transactions out of the total number of initiated transactions. This is a key metric to measure the system's ability of conflict resolution in highly concurrent \textsf{CSTx} scenarios, because in the worst case, only the first \textsf{CSTx} in a block may be successful and all other \textsf{CSTx}s are aborted due to concurrent data conflicts. (ii) \textit{Transaction throughput.} The number of transactions per second (TPS) successfully processed by the system. This is also an important measure of concurrency capacity. (iii) \textit{Transaction latency.} The time taken by the user from the initiation of a \textsf{CSTx} to the final successful write to the ledger. (iv) \textit{CPU \& memory utilization.} The percentage of CPU and memory resources utilized by the system during operation.

\subsection{Evaluation Results}
\textbf{High concurrency testing.} The focus of our work is on resolving highly concurrent conflicts in \textsf{CSTx}s. Thus, we first present the system's performance without enabling the homomorphic encryption module for privacy preservation. To evaluate HiCoCS and the baselines' ability to handle transaction conflicts, we measure the transaction success rate and average throughput by actively and passively introducing conflicts during testing. We then evaluate the energy efficiency of concurrent transactions and the dynamic performance under varying numbers of highly concurrent transactions.

\begin{figure}[t]
	\vspace{-0.2cm}
	\centering
	\subfigure[Transaction success rate (TSR).]{
		\begin{minipage}[t]{0.49\linewidth}
			\centering
			\includegraphics[width=1.6in]{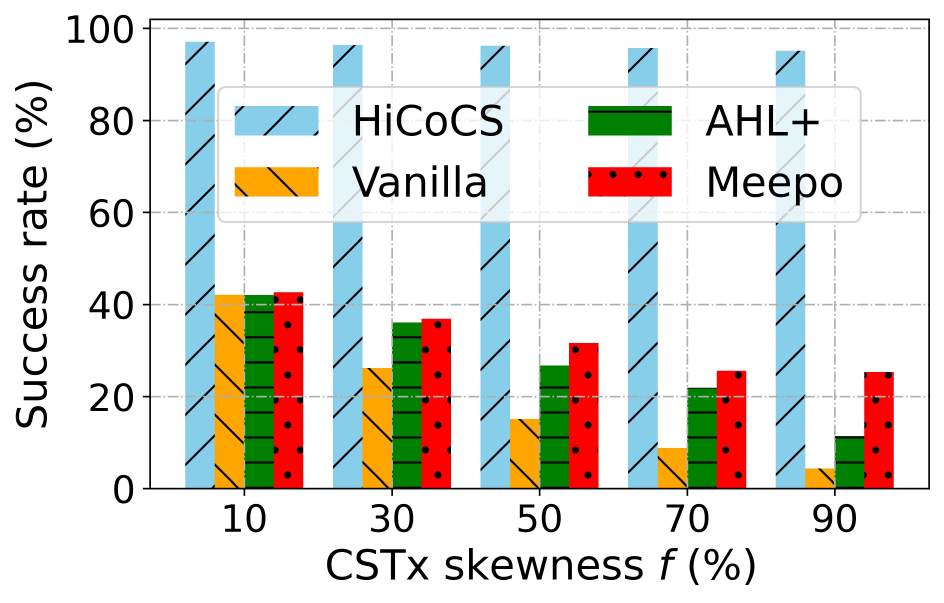}
			\label{TSR-Skewness}
		\end{minipage}%
	}%
	\subfigure[Transaction throughput (TPS).]{
		\begin{minipage}[t]{0.49\linewidth}
			\centering
			\includegraphics[width=1.6in]{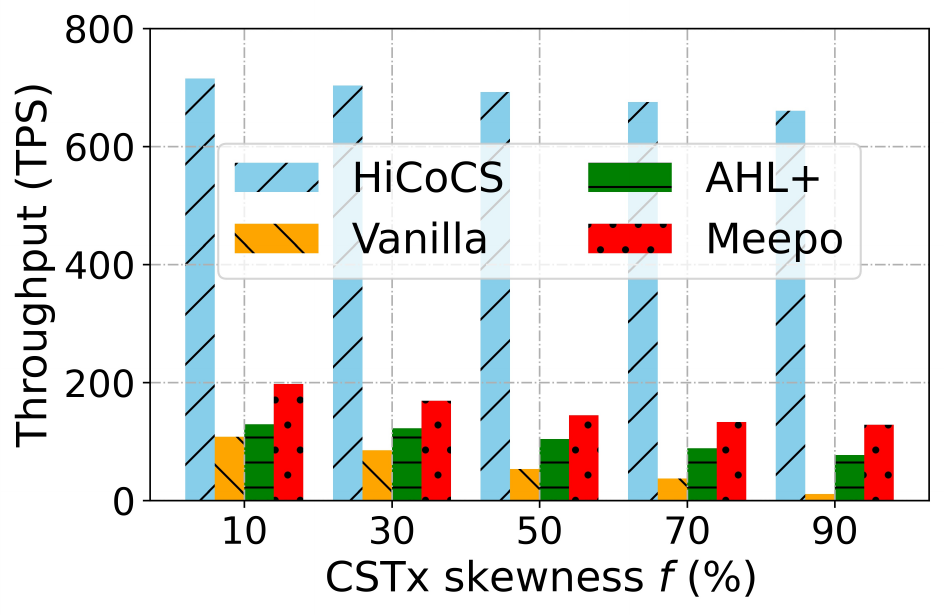}
			\label{TPS-Skewness}
		\end{minipage}
	}%
	\centering
	\vspace{-0.2cm}
	\caption{Comparison of concurrency performance under varying skewness $f$.}
	\label{active}
	\vspace{-0.3cm}
\end{figure}

\begin{figure}[t]
	\centering
	\subfigure[Transaction success rate (TSR).]{
		\begin{minipage}[t]{0.49\linewidth}
			\centering
			\includegraphics[width=1.6in]{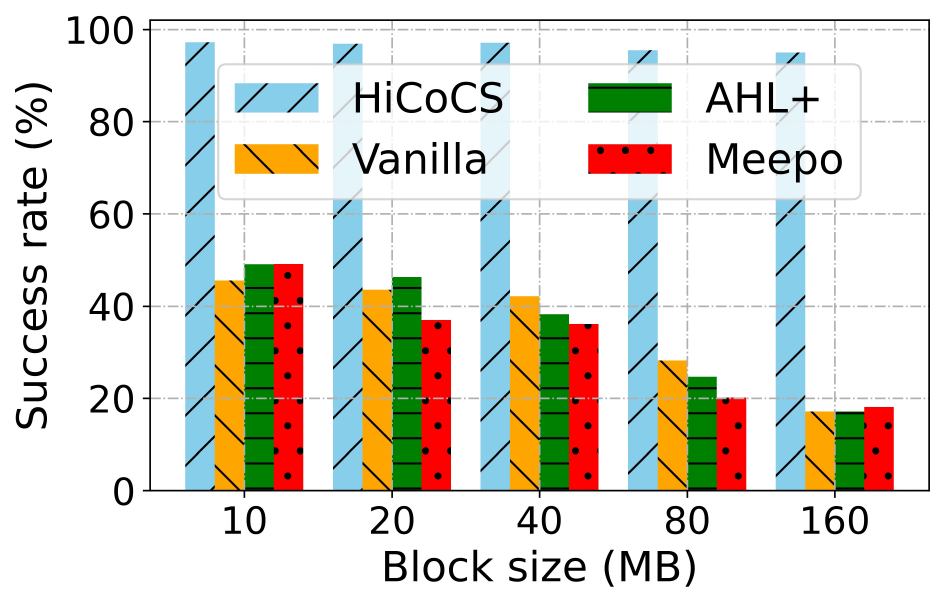}
			\label{TSR-BlockSize}
		\end{minipage}%
	}%
	\subfigure[Transaction throughput (TPS).]{
		\begin{minipage}[t]{0.49\linewidth}
			\centering
			\includegraphics[width=1.6in]{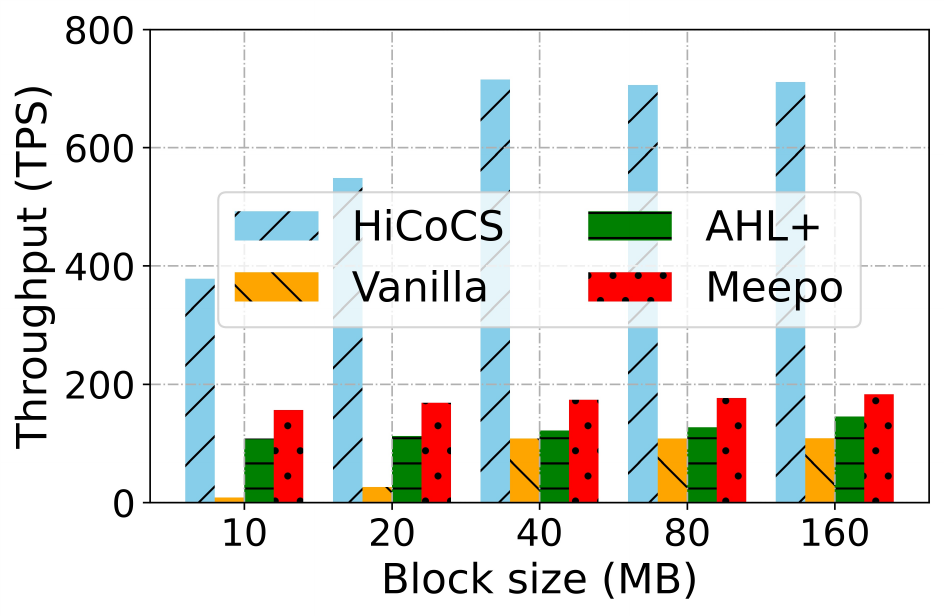}
			\label{TPS-BlockSize}
		\end{minipage}
	}%
	\centering
	\vspace{-0.2cm}
	\caption{Comparison of concurrency performance under varying block sizes.}
	\label{Passive}
	\vspace{-0.3cm}
\end{figure}

\begin{table}[t]
	\begin{center}
		\caption{Comparison of Concurrent Energy Efficiency (CPU \& Memory Utilizations, Unit: \%).}
		\label{CMU}	
		\vspace{-0.5cm}
		\resizebox{1\columnwidth}{!}{	
			\begin{threeparttable}
				\begin{tabular}{cccccc}
					\toprule
					Skewness &10 &30 &50 &70 &90  \\
					\midrule
					\textbf{HiCoCS} &\textbf{55.6} \& \textbf{58.2} &\textbf{62.5} \& \textbf{66.4} &\textbf{75.4} \& \textbf{71.2} &\textbf{87.9} \& \textbf{76.8} &\textbf{90.3} \& \textbf{80.1}\\
					
					Vanilla &72.1 \& 75.0 &78.9 \& 82.4 &87.9 \& 88.0 &95.4 \& 91.3 &98.5 \& 93.7\\
					
					AHL+ &68.5 \& 71.2 &75.6 \& 79.1 &84.5 \& 85.2 &93.1 \& 89.8 &97.0 \& 92.1\\
					
					Meppo &64.1 \& 67.9 &71.8 \& 76.1 &81.3 \& 82.9 &90.4 \& 87.2 &94.0 \& 90.3\\
					\bottomrule
				\end{tabular}
				
				
			\end{threeparttable}
		}
	\end{center}	
	\vspace{-0.3cm}
\end{table}

\textit{1) Active test results:} We fix the block size at 40 MB and change the transaction skewness $f$ to perform transaction conflict testing. Fig. \ref{active} shows the influence of skewness $f$ on concurrent transactions. The results of Fig. \ref{TSR-Skewness} indicate that all baseline transaction success rates decrease as skewness $f$ increases. However, HiCoCS experiences minimal decrease, stabilizing above 95\%. HiCoCS improves TSR by an average factor of 2.2 to 8.1 compared to the baselines. We analyze that as skewness $f$ increases, HiCoCS generates more composite keys for the intermediaries of high-frequency services, thereby minimizing the increase in transaction conflict rates. In contrast, the baselines experience more transaction aborts due to competitive conditions. Fig. \ref{TPS-Skewness} further illustrates the comparison of average throughput. It is observed that the throughput of all schemes decreases as skewness $f$ increases. However, HiCoCS improves in TPS by an average of 3.5 to 20.2 times compared to the baselines. This occurs because the baselines experience a significant number of transaction retries as transaction aborts increase. HiCoCS effectively mitigates concurrency conflicts by using virtual sub-brokers.

\textit{2) Passive  test results:} We set the skewness $f$ to 10\% and vary the block size for passive conflict testing. Fig. \ref{Passive} illustrates the impact of different block sizes. Fig. \ref{TSR-BlockSize} demonstrates that the baseline transaction success rates decrease as block size increases. However, similar to the active test, HiCoCS experiences almost no decrease. On average, HiCoCS improves TSR by a factor of 2.1 to 2.3 compared to the baselines. This aligns with expectations, as larger blocks increase the probability of \textsf{CSTx} conflicts; however, HiCoCS pre-processes concurrent transactions to minimize conflict occurrences, effectively avoiding most of them. Fig. \ref{TPS-BlockSize} illustrates that TPS increases across all schemes as block size increases. This is consistent with expectations, as larger blocks reduce the frequency of network communications. On average, HiCoCS improves throughput by a factor of 2.5 to 15.7 compared to the baselines.

\begin{figure}[t]
	\centering
	\subfigure[Throughput (TPS).]{
		\begin{minipage}[t]{0.45\linewidth}
			\centering
			\includegraphics[width=1.6in]{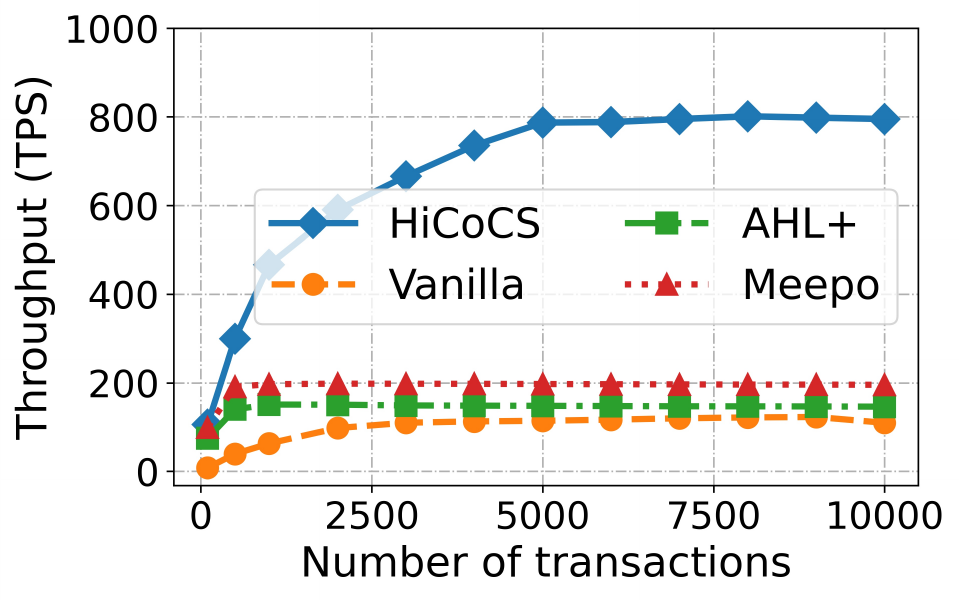}
			\label{PeakTPS-TxNum}
			\vspace{-0.6cm}
		\end{minipage}%
	}%
	\hspace{0.3cm}
	\subfigure[Latency (ms).]{
		\begin{minipage}[t]{0.45\linewidth}
			\centering
			\includegraphics[width=1.6in]{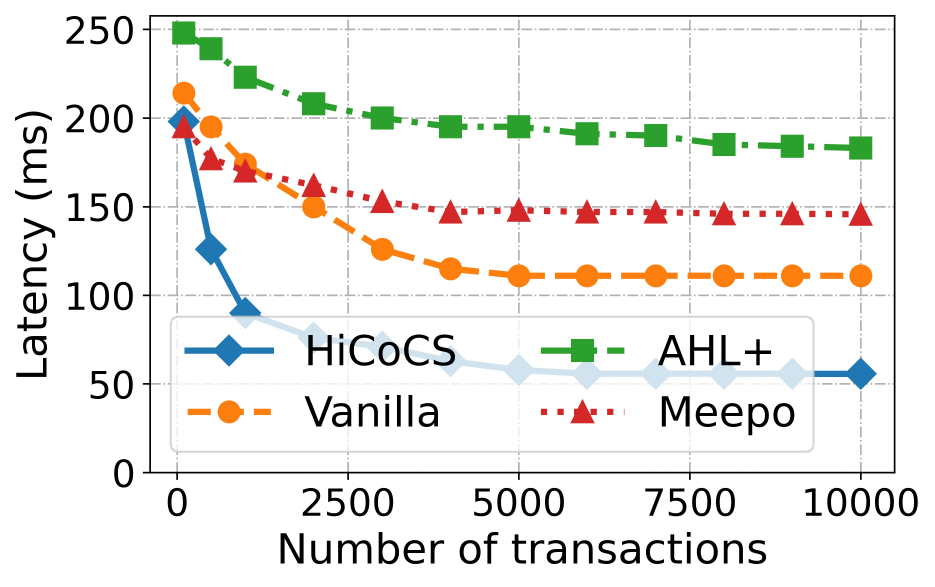}
			\label{Latency-TxNum}
			\vspace{-0.6cm}
		\end{minipage}
	}%
	
	
	\subfigure[CPU utilization (\%).]{
		\begin{minipage}[t]{0.45\linewidth}
			\centering
			\includegraphics[width=1.6in]{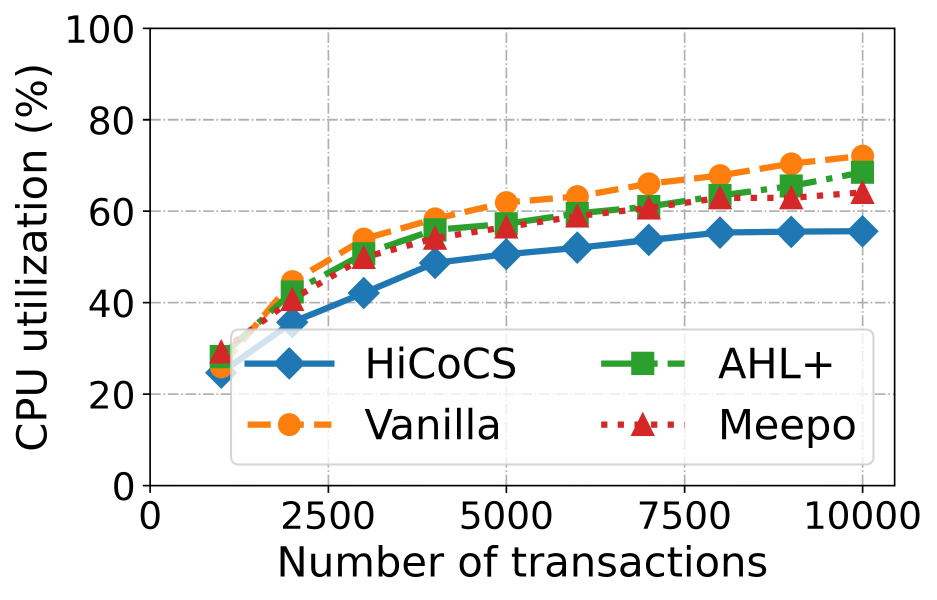}
			\label{CPU-TxNum}
			\vspace{-0.6cm}
		\end{minipage}%
	}%
	\hspace{0.3cm}
	\subfigure[Memory utilization (\%).]{
		\begin{minipage}[t]{0.45\linewidth}
			\centering
			\includegraphics[width=1.6in]{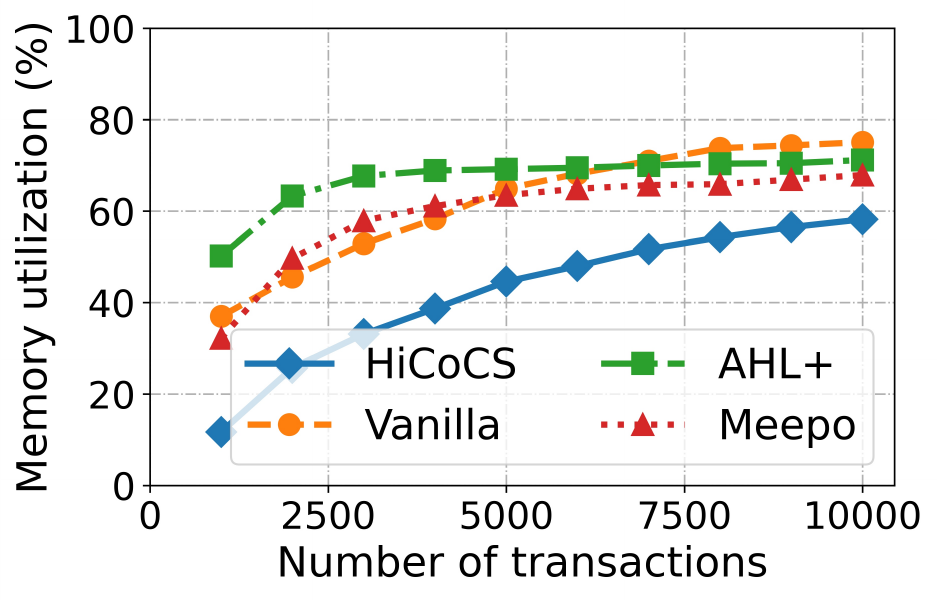}
			\label{Memory-TxNum}
			\vspace{-0.6cm}
		\end{minipage}
	}%
	\centering
	\vspace{-0.2cm}
	\caption{Dynamic performance comparison for varying numbers of transactions.}
	\label{Dynamic-P}
	\vspace{-0.3cm}
\end{figure}

\textit{3) Concurrent energy efficiency results:} Table \ref{CMU} presents the average CPU and memory usage of different schemes for each shard pair under the load of 10,000 highly concurrent transactions, with a block size of 40 MB and varying skewness. Overall, the resource overhead of all schemes increases as skewness increases. However, HiCoCS reduces CPU and memory consumption by an average of 8.0\% to 14.8\% and 12.9\% to 18.3\%, respectively, compared to the baselines. This reflects, on the one hand, the efficiency advantage of HiCoCS in handling highly concurrent transaction conflicts without overloading due to frequent transaction reissuance. On the other hand, it also demonstrates the optimization of resource allocation enabled by composite key reuse in HiCoCS.

\textit{4) Dynamic performance results:} We further compare the dynamic performance of HiCoCS and the baselines under varying numbers of highly concurrent transactions. The skewness $f$ is set to 10\%, and the block size is fixed at 40 MB. Fig. \ref{Dynamic-P} shows the variation in transaction throughput, latency, and CPU \& memory utilization. The results of Fig. \ref{PeakTPS-TxNum} indicate that the throughput of HiCoCS only begins to level off when transaction volume is larger compared to the baselines, demonstrating HiCoCS's stronger concurrent transaction processing capability. Fig. \ref{Latency-TxNum} shows that the average transaction latency of HiCoCS is 43.9\%-62.0\% lower than that of the baselines. This reflects HiCoCS's high efficiency in batch processing transactions. Fig. \ref{CPU-TxNum} and \ref{Memory-TxNum} show that CPU and memory utilization in HiCoCS are 12.4\%-18.0\% and 31.6\%-38.5\% lower than those of the baselines, respectively. This demonstrates HiCoCS's superiority in terms of resource utilization.

\begin{figure}[h]
	\centering
	\subfigure[CPU utilization (\%).]{
		\begin{minipage}[t]{0.45\linewidth}
			\centering
			\includegraphics[width=1.6in]{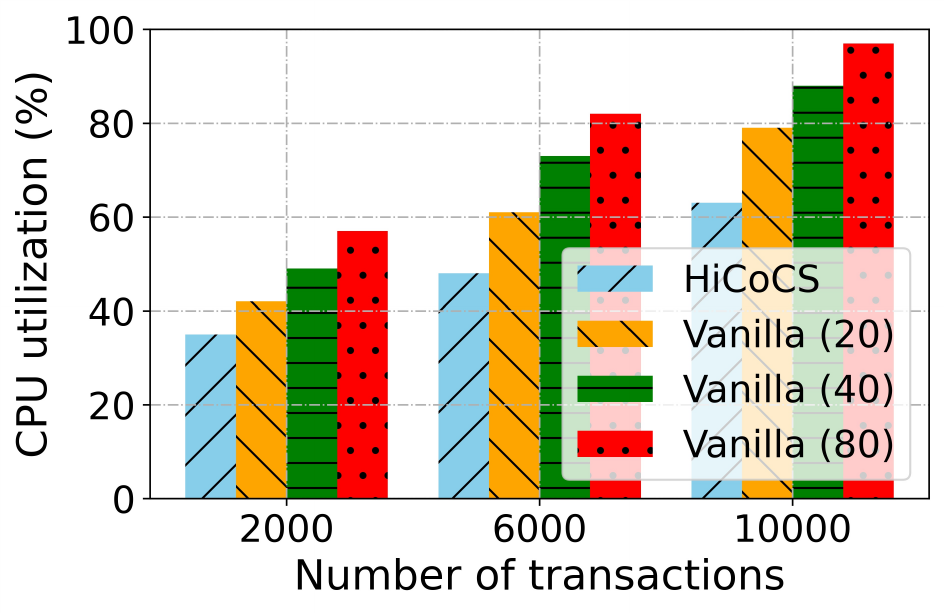}
			\label{CPU-VS}
			\vspace{-0.6cm}
		\end{minipage}%
	}%
	\hspace{0.3cm}
	\subfigure[Memory utilization (\%).]{
		\begin{minipage}[t]{0.45\linewidth}
			\centering
			\includegraphics[width=1.6in]{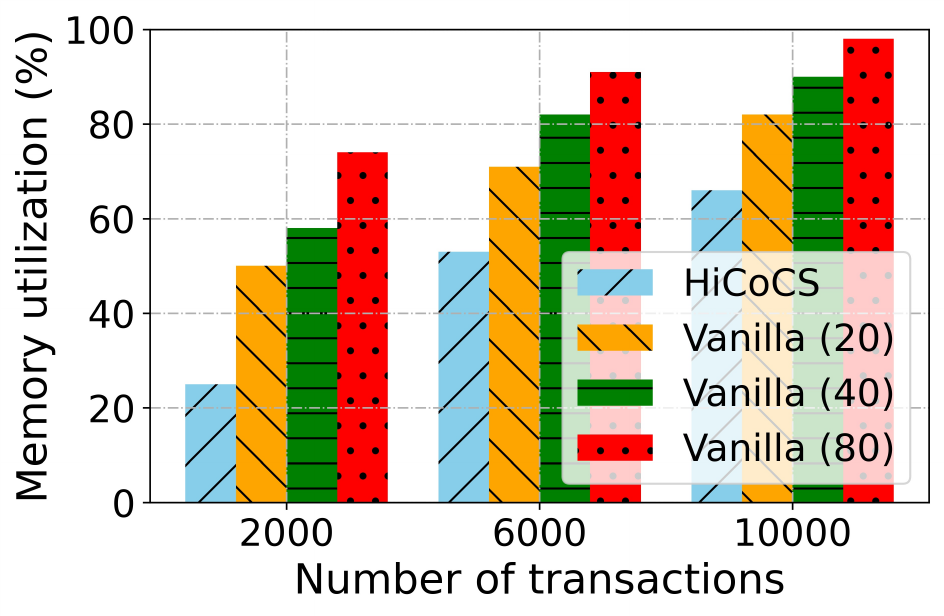}
			\label{Mem-VS}
			\vspace{-0.6cm}
		\end{minipage}
	}%
	
	
	\subfigure[Throughput (TPS).]{
		\begin{minipage}[t]{0.45\linewidth}
			\centering
			\includegraphics[width=1.6in]{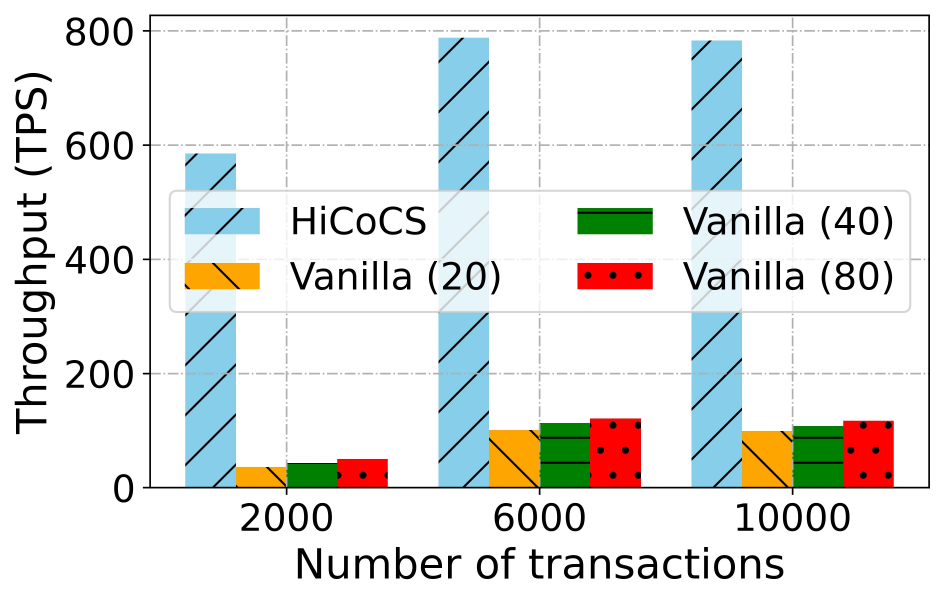}
			\label{TPS-VS}
			\vspace{-0.6cm}
		\end{minipage}%
	}%
	\hspace{0.3cm}
	\subfigure[Latency (ms).]{
		\begin{minipage}[t]{0.45\linewidth}
			\centering
			\includegraphics[width=1.6in]{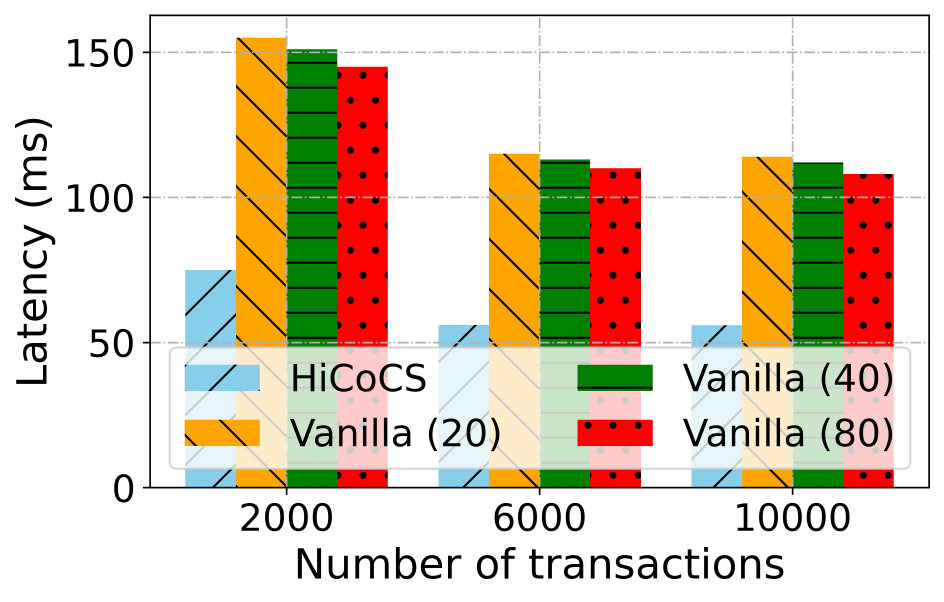}
			\label{Latency-VS}
			\vspace{-0.6cm}
		\end{minipage}
	}%
	\centering
	\vspace{-0.2cm}
	\caption{A quantitative comparison of system overhead and performance between using composite keys and adding intermediaries.}
	\label{CK-VS}
	\vspace{-0.3cm}
\end{figure}

\textbf{Composite key testing.} First, we quantitatively test composite keys to compare the system overhead and performance when using composite keys versus adding new intermediaries. Next, we tested the memory consumption and query time before and after composite key reuse to verify its effectiveness.

\textit{1) Quantitative test and analysis:} The default setting for the intermediary group $\mathcal{G}$ between each shard pair consists of 20 accounts providing cross-shard trading services. We vary the vanilla version's size of $\mathcal{G}$ to 40 and 80 to compare it with HiCoCS. The skewness $f$ is set to 30\%, and HiCoCS generates and reuses composite keys based on conflicting transactions. Fig. \ref{CPU-VS} and \ref{Mem-VS} present the comparative results of CPU and memory utilization for varying numbers of transactions. The results indicate that HiCoCS reduces CPU and memory usage by an average of 19.4\%-38.4\% and 31.6\%-46.9\%, respectively, compared to the method of adding intermediaries. We analyze this because the effect of adding a virtual sub-broker is nearly equivalent to that of adding a new intermediary. However, the former incurs minimal memory overhead, whereas the latter introduces significant overhead for authentication, account creation \& maintenance, and transaction processing scheduling. Thus, this test validates that HiCoCS's approach of using virtual sub-brokers is more efficient than adding new intermediaries. Fig. \ref{TPS-VS} and \ref{Latency-VS} present the comparative results for average transaction throughput and latency. The results in Fig. \ref{TPS-VS} indicate that adding intermediaries can slightly improve throughput. This is because contention conflicts among intermediaries remain severe under highly concurrent transactions. HiCoCS demonstrates an average improvement of 7.3 to 9.7 times in throughput compared to the method of adding intermediaries. The results in Fig. \ref{Latency-VS} indicate that adding intermediaries slightly reduces latency due to the reduction in conflict locking time. HiCoCS reduces latency by an average of 48.5\% to 51.3\% compared to them.

\textit{2) Reuse effectiveness test:} The number of composite keys primarily affects system memory usage. Fig. \ref{Mem-U-TxNum} presents the results of memory utilization before and after composite key optimization, showing that memory utilization gradually increases and stabilizes as the number of \textsf{CSTx} increases. However, memory utilization decreases significantly after optimization. The average memory utilization of the scheme with composite key reuse enabled is reduced by 31.6\%. The time taken for the intermediary to query the \verb|ComKeyPool| before and after composite key reuse optimization is shown in Fig. \ref{Q-Time}. As the number of \textsf{CSTx} increases, the time taken by the intermediary to query the \verb|ComKeyPool| also increases. This is because more composite keys are created within \verb|ComKeyPool|, expanding the range of the intermediary's fuzzy query and consequently increasing query time. After composite key reuse optimization, the average query time for \textsf{CSTx} is reduced by 8.3\%.

\begin{figure}[t]
	\vspace{-0.2cm}
	\centering
	\subfigure[Memory utilization (\%).]{
		\begin{minipage}[t]{0.49\linewidth}
			\centering
			\includegraphics[width=1.6in]{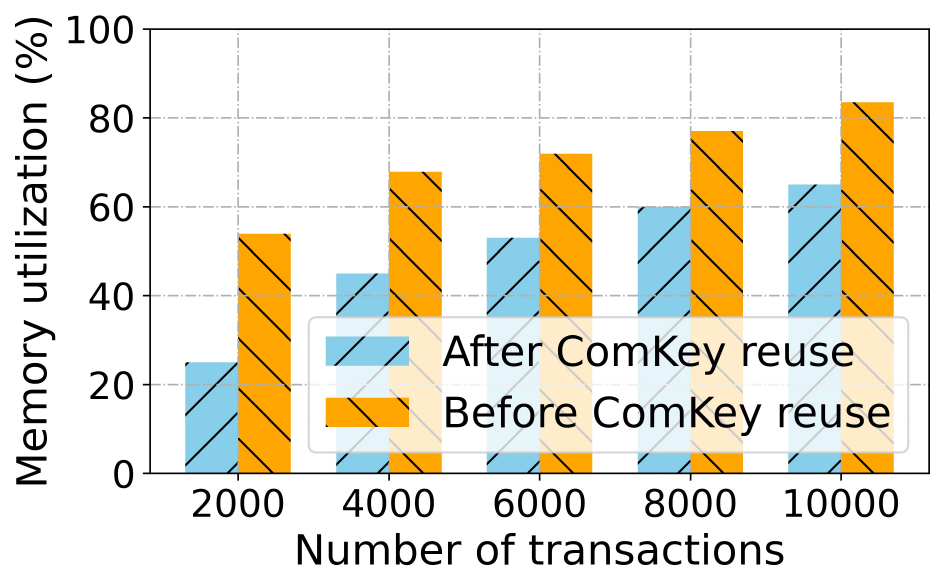}
			\label{Mem-U-TxNum}
			\vspace{-0.5cm}
		\end{minipage}%
	}%
	\subfigure[Time of querying \texttt{ComKeyPool}.]{
		\begin{minipage}[t]{0.49\linewidth}
			\centering
			\includegraphics[width=1.6in]{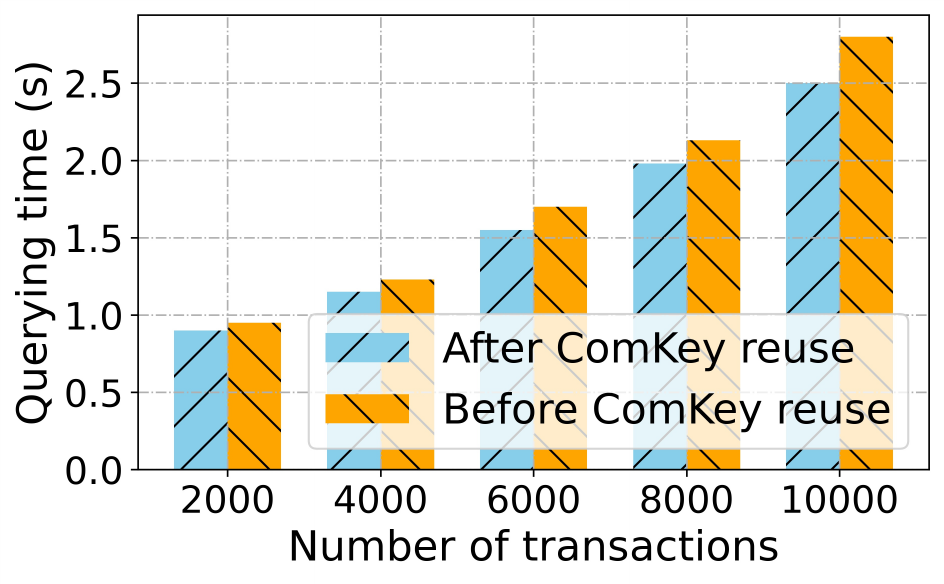}
			\label{Q-Time}
			\vspace{-0.5cm}
		\end{minipage}
	}%
	\centering
	\vspace{-0.2cm}
	\caption{Performance comparison before and after composite key reuse optimization in HiCoCS.}
	\vspace{-0.3cm}
\end{figure}

\textbf{Privacy preservation-enabled testing.} We first evaluate the impact on system performance after enabling the privacy-preserving mechanism, followed by an evaluation of the performance changes after scaling the sharded network. Finally, we evaluate the results of ciphertext computation.

\textit{1) Performance impact of FHE:} Fig. \ref{PP1} shows the performance change of the system before and after adding the privacy-preserving mechanism (i.e., FHE). The results of Fig. \ref{TPS-3Type} HiCoCS throughput decreases by an average of 19.5\% due to the series of encryption operations after FHE is enabled. However, compared to the vanilla version, throughput still improves by 5.7 times. HiCoCS (with FHE) throughput reaches 702.3 TPS when the number of transactions reaches 8,000, and then stabilizes due to machine performance limitations. Fig. \ref{Latency-3Type} shows the impact of privacy preservation on system latency. The maximum latency of HiCoCS (with FHE) is 203.2 ms, and the average latency is 118.7 ms, representing a 32.6\% increase compared to HiCoCS (without FHE), which has an average latency of 80.0 ms. However, compared to the vanilla version, latency in HiCoCS (with FHE) is still reduced by an average of 13.2\%. The average processing latency of the system decreases as transaction volume increases because the number of concurrent transactions processed per unit of time increases, and latency stabilizes after the system reaches its processing limit.

\begin{figure}[t]
	\centering
	\subfigure[Transaction throughput (TPS).]{
		\begin{minipage}[t]{0.49\linewidth}
			\centering
			\includegraphics[width=1.6in]{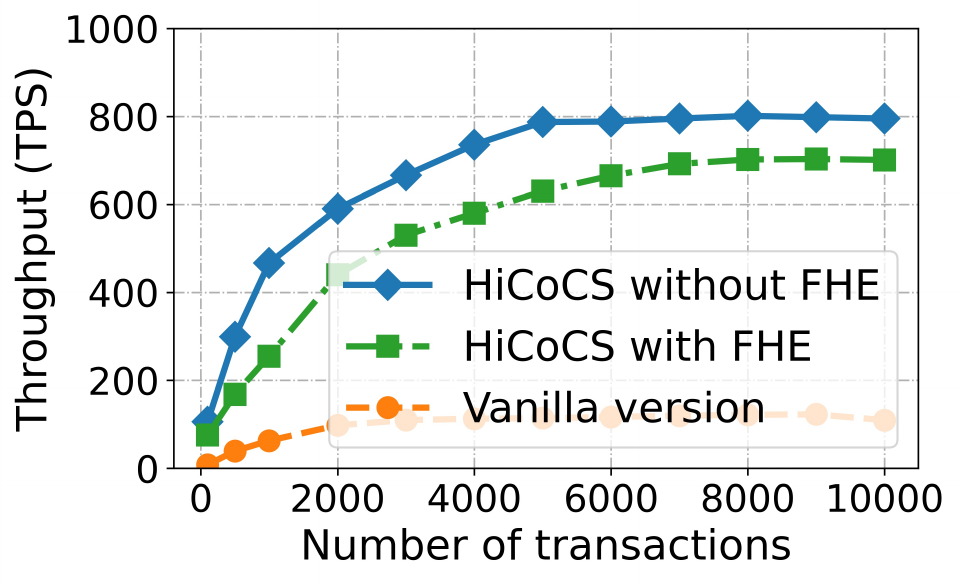}
			\label{TPS-3Type}
		\end{minipage}%
	}%
	\subfigure[Transaction latency (ms).]{
		\begin{minipage}[t]{0.49\linewidth}
			\centering
			\includegraphics[width=1.6in]{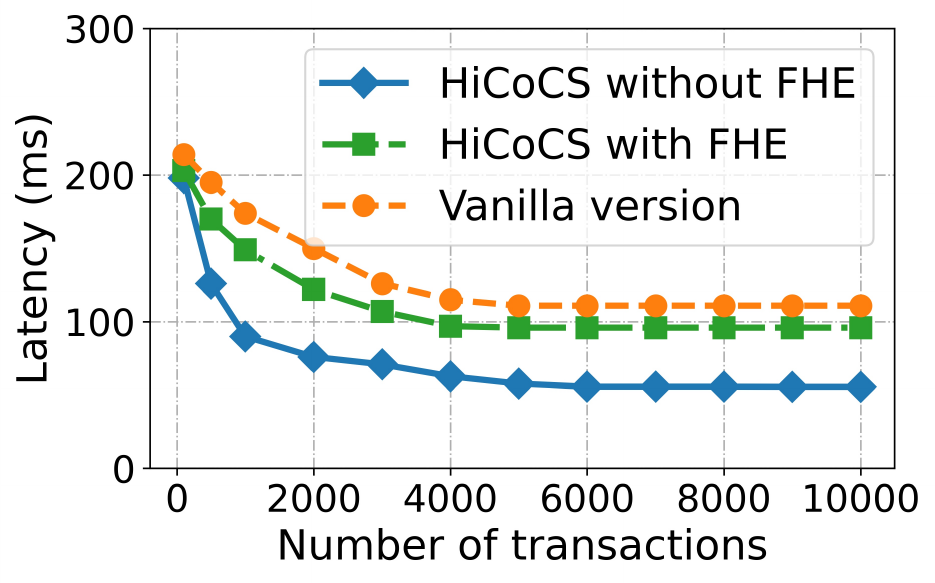}
			\label{Latency-3Type}
		\end{minipage}
	}%
	\centering
	\caption{Performance comparison of HiCoCS before and after enabling privacy-preserving mechanism (under the 16-shard, 64-node network).}
	\label{PP1}
\end{figure}

\begin{figure}[t]
	\vspace{-0.2cm}
	\centering
	\subfigure[Transaction throughput (TPS).]{
		\begin{minipage}[t]{0.49\linewidth}
			\centering
			\includegraphics[width=1.6in]{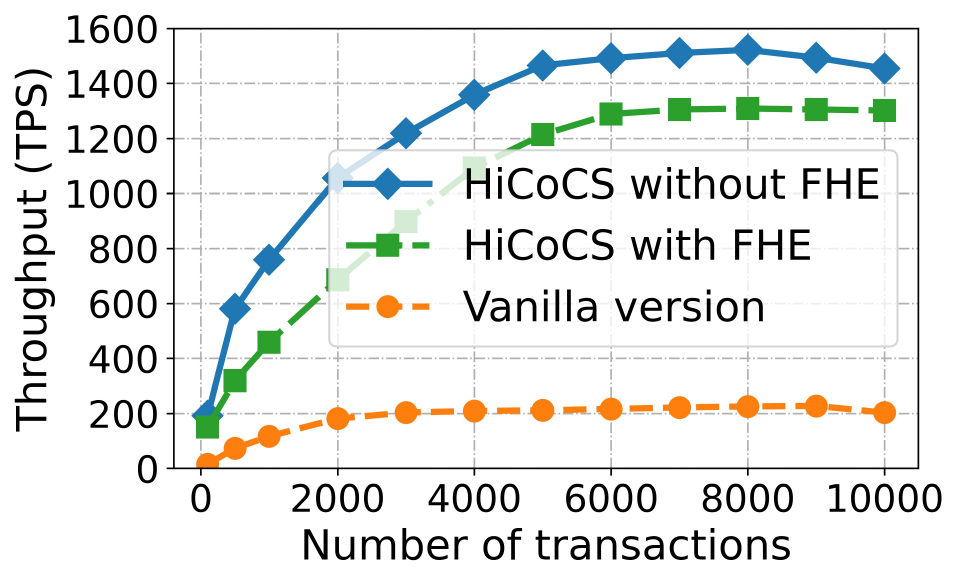}
			\label{TPS-3Type2}
			\vspace{-0.5cm}
		\end{minipage}%
	}%
	\subfigure[Transaction latency (ms).]{
		\begin{minipage}[t]{0.49\linewidth}
			\centering
			\includegraphics[width=1.6in]{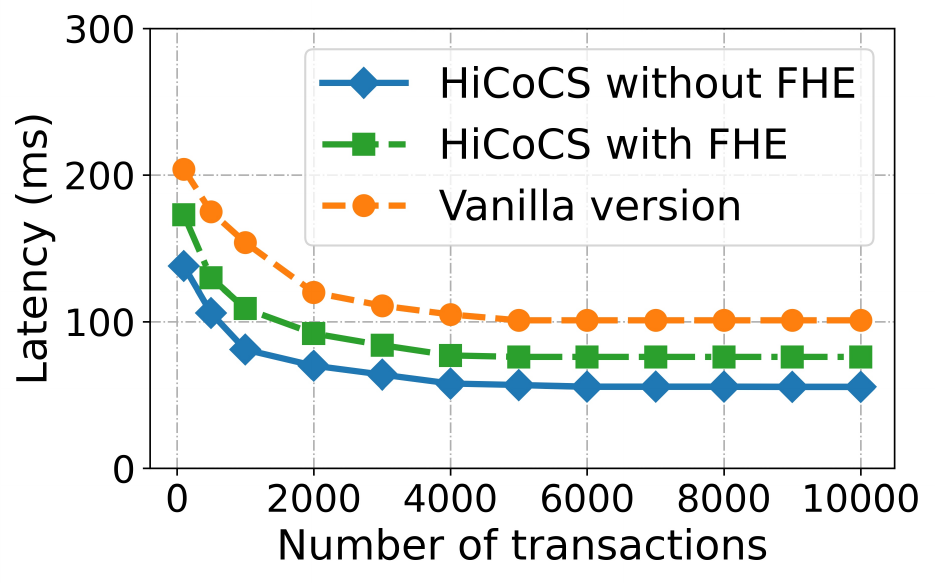}
			\label{Latency-3Type2}
			\vspace{-0.5cm}
		\end{minipage}
	}%
	\centering
	\vspace{-0.2cm}
	\caption{Performance comparison of HiCoCS before and after enabling privacy-preserving mechanism (under the 32-shard, 128-node network).}
	\label{PP2}
	\vspace{-0.3cm}
\end{figure}

\textit{2) Network scaling test:} The default test environment is conducted in a 16-shard, 64-node network to avoid reaching the machine performance bottleneck. To evaluate the scalability of the overall HiCoCS scheme, we present the scaling results in a 32-shard, 128-node network in Fig \ref{PP2}. In this larger network, as shown in Fig. \ref{TPS-3Type2}, compared to previous test results in a 16-shard, 64-node network, the throughput exhibits a nearly twofold linear growth trend. This indicates that HiCoCS demonstrates good scalability with all building blocks enabled, processing more transactions securely in parallel. Fig. \ref{Latency-3Type2} shows transaction latency in the 32-shard, 128-node network, and the results indicate that HiCoCS latency is reduced by an average of 21.3\% compared to the 16-shard, 64-node network. We analyze that although increasing the number of shards distributes the network's transaction load and speeds up transaction processing within each shard, it also increases communication latency due to the greater number of cross-shard transactions. HiCoCS is more resource-efficient than the vanilla version, leading to a more significant reduction in transaction latency.

\textit{3) Ciphertext calculation results:} To ensure the accuracy of the ciphertext calculation, we evaluate the gap between the approximate CKKS ciphertext homomorphic computations and the actual calculated values. We randomly initiated 10,000 \textsf{CSTx}s, with their transaction amounts encrypted using CKKS, before entering the ciphertext summation. Finally, we decrypt the summation ciphertext results and compare them with the summation results of the original data. The results show that the error rate is within $10^{-5}$ compared to the actual transaction amounts. In practice, this error can be ignored or offset by the transaction fee.

%% file: 7-Conclusion.tex
\section{Conclusion} \label{conclusion}
This paper proposes HiCoCS, the first Hyperledger Fabric-based implementation of cross-shard transaction middleware featuring high concurrency and privacy preservation. We utilize composite keys to build virtual sub-brokers for intermediaries in the vanilla version to mitigate concurrent transaction conflicts. We also consider composite key reuse to their number and lower system overhead. HiCoCS utilizes fully homomorphic encryption to ensure the privacy preservation of intermediaries in cross-shard transactions. Our evaluation of the developed prototype demonstrates that HiCoCS outperforms the vanilla version and state-of-the-art schemes across all metrics. In future work, we plan to enhance the generality of HiCoCS and extend this work to more blockchains and a wide range of Web3 applications.